\documentclass[11pt,a4paper]{article}

\usepackage{jheppub}
\usepackage{amsthm}

\usepackage{multirow, graphicx,amssymb,url,mathrsfs,amsmath}
\usepackage{wrapfig,boxedminipage,setspace,subfigure,epsfig}
\usepackage{amsxtra,amstext,latexsym,dsfont,amsfonts}
\usepackage{color,eucal}
\usepackage[dvipsnames]{xcolor}
\usepackage{float}
\usepackage{slashed,comment}
\usepackage{kotex}
\usepackage{tensor}

\usepackage{indentfirst}

\newtheorem{theorem}{Theorem}[section]

\newtheorem{lemma}[theorem]{Lemma}

\newcommand{\bra}[1]{\mbox{$\langle #1 |$}}
\newcommand{\ket}[1]{\mbox{$| #1 \rangle$}}

\newcommand{\Tr}{{\rm Tr}\,}

\newcommand{\eg}{{\it e.g.}}
\newcommand{\ie}{{\it i.e.}}

\title{Squashed Entanglement from Generalized Rindler Wedge}

\author[a]{Xin-Xiang Ju,}
\author[a]{Bo-Hao Liu,}
\author[b]{Wen-Bin Pan,}
\author[a,c]{Ya-Wen Sun}
\author[a]{and Yuan-Tai Wang}

\emailAdd{juxinxiang21@mails.ucas.ac.cn}
\emailAdd{liubohao16@mails.ucas.ac.cn}
\emailAdd{panwenbin18@mails.ucas.ac.cn}
\emailAdd{yawen.sun@ucas.ac.cn}
\emailAdd{wangyuantai19@mails.ucas.ac.cn}

\affiliation[a]{School of Physical Sciences, University of Chinese Academy of Sciences, Zhongguancun east road 80, Beijing 100190, China}
\affiliation[b]{Institute of High Energy Physics, Chinese Academy of Sciences,\\19B Yuquan Road, Shijingshan District, Beijing 100049, China}
\affiliation[c]{Kavli Institute for Theoretical Sciences, University of Chinese Academy of Sciences, Beijing 100049, China}

\abstract{We investigate the bipartite and multipartite quantum entanglement structure of gravitational subsystems and in the dual holographic field theory based on the generalized Rindler wedge formalism. We deduce a separation theorem, which asserts that for subregions satisfying a certain geometric condition, the bipartite/multipartite squashed entanglement or the conditional entanglement of multipartite information vanishes, indicating that these subregions represent separable states with no quantum entanglement among them. We interpret this fact from the observer perspective in gravity and show how to probe the entanglement structure further in this framework by introducing a time cutoff in the gravitational spacetime. We also present the corresponding dual boundary field theory interpretation.}

\begin{document}
\maketitle

\section{Introduction}
\noindent
The profound significance of quantum entanglement in gravitational and holographic theories has been acknowledged for a long time. This recognition traces back to the exploration of the Bekenstein-Hawking entropy associated with black holes \cite{Bekenstein:1973ur,Bardeen:1973gs,Bekenstein:1974ax,Hawking:1975vcx,Eisert:2008ur,Jacobson:2003wv}, an endeavor that subsequently lead to the formulation of the holographic principle \cite{Susskind_1995} and laid important foundation for the establishment of the AdS/CFT correspondence \cite{Maldacena_1999}. 
With the proposal of the Ryu-Takayanagi (RT) formula \cite{Ryu_2006}, which provides a geometric dual for entanglement entropy in holographic field theories, the profound idea that ``entanglement builds geometry" \cite{Rangamani_2017,VanRaamsdonk:2010pw} took shape. This includes investigations into tensor networks \cite{swingle2012constructing,Vidal_2007,vidal2010entanglement}, the ER=EPR proposal \cite{Maldacena_2013}, and other related studies, which have contributed to unveiling the entanglement structures in gravitational and holographic theories.

The RT formalism states that bulk degrees of freedom within the entanglement wedge, which is the region bounded by extremal surfaces, correspond to degrees of freedom of the corresponding boundary subregions. Thus entanglement wedges define consistent gravitational subsystems. {The subregion-subregion duality and various quantum information properties based on the entanglement wedge have been studied extensively. In this work, we study the entanglement properties, especially bipartite and multipartite entanglement, of more general shapes of gravitational subregions: the Rindler convex subregions. Given the non-local nature of gravity, gravitational subsystems cannot be arbitrarily defined. Rindler convex subregions are well-defined gravitational subregions based on the following three criteria.}

 {I. The subalgebra-subregion duality. In \cite{Leutheusser:2022bgi}, the subregion subregion duality has been generalized to the subalgebra-subregion duality, where more general shapes of bulk subregions could be associated with a special type of boundary time band subalgebra and this defines well-defined gravitational subregions. For this boundary time band subalgebra, the bulk subregion has to be the complement of a geodesically convex subregion on the time-reflection symmetric Cauchy slice. The convex bulk subregion therefore corresponds to the commutant of a time-band algebra. As the existence of a subalgebra implies well-defined gravitational subsystems, the convex subregions are indeed well-defined subregions. Note that the geodesic convexity here required in the subalgebra subregiond duality is equivalent to the Rindler convexity condition, whose definition we will explain later, in the special cases that were studied in \cite{Leutheusser:2022bgi}.}
 
 {II. The generalized Rindler wedge formalism.  In \cite{Ju:2023bjl}, another way to define gravitational subsystems is proposed, where instead of using extremal surfaces, we introduced a set of accelerating observers and associated well-defined subregions based on whether observers can precisely access those regions. A Rindler convexity condition is required on the shape of the enclosing surface, and the subsystems that could be accessible by well-defined accelerating observers are named generalized Rindler wedges. 
This proposal for defining gravitational subsystems aligns with the subregion subalgebra duality in \cite{Leutheusser:2022bgi}. These GRW gravitational subsystems are associated with a special type of Type III von Neumann subalgebra \cite{Leutheusser:2022bgi}, and this GRW formalism provides a physical explanation for the convexity constraint in \cite{Leutheusser:2022bgi} when defining subalgebra in the subregion. As the existence of a subalgebra serves as a guideline in defining gravitational subsystems, this fact validates the consistency of these generalized gravitational subregions as well-defined gravitational subsystems.}

 {III. The Bousso Pennington generalized entanglement wedge \cite{Bousso:2022hlz, {Balasubramanian:2023dpj}}.
{From \cite{Bousso:2022hlz} and \cite{Balasubramanian:2023dpj}, another consistent way of defining gravitational subsystems is introduced by defining a gravitational entanglement wedge for a given gravitational subregion. The results indicate that the entanglement entropy of a given gravitational subregion is given by the minimal area of the surface that covers it, bounding the spatial part of the generalized entanglement wedge.  {In section 2, we will prove that in a large class of spacetimes,} Rindler convex subregions are generalized entanglement wedges in the Bousso Pennington formalism \cite{Bousso:2022hlz}. For convex subregions as we are investigating, this area is exactly the surface area of this convex region. Consequently, our results in \cite{Ju:2023bjl} could also be derived in an alternative way as in \cite{Balasubramanian:2023dpj} and \cite{Bousso:2022hlz}.}}

Being a consistent subsystem, the entanglement entropy\footnote{ {The Type III von Neumann subalgebra could be promoted to a Type II subalgebra so that the entanglement entropy could be defined as done in \cite{Jensen:2023yxy} with a local modular Hamiltonian \cite{Ju:2025mvz}.}} for generalized Rindler wedges could be calculated from the gravitational path integral, which is proportional to the area of the Rindler convex surface. Also, its holographic dual on the boundary field theory side as the entanglement entropy of a specific state $\tilde \rho$, where corresponding long-range entanglement is eliminated, is also established \cite{Czech_2015, Ju:2023bjl}. The complementary subregion of the generalized Rindler wedge on the Cauchy slice is named the Rindler convex region, whose entanglement entropy is also proportional to the area of the Rindler convex surface as the whole state is a pure state.

{Compared to the entanglement wedge, which relies on extremal surfaces to define consistent subregions, the generalized Rindler wedge formalism imposes a weaker requirement, only necessitating Rindler-convexity of the enclosing surface. This condition provides more flexibility in selecting the shape of the surfaces. {Also the quantum states dual to these Rindler convex subregions have been shown to be well-defined states with certain IR entanglement structures removed consistently \cite{Ju:2023bjl}. Thus this} improves our ability to analyze intricate gravitational entanglement structures that the RT formula cannot address.}

However, entanglement entropy only measures the amount of entanglement of a subsystem with its complement without delivering more precise information about which components are engaged in and contribute to the entanglement. One could investigate entanglement between subregions to reveal fine entanglement structures, rather than merely focusing on the entanglement entropy. A challenge thus arises when analyzing the entanglement between subregions in gravitational and holographic theories, i.e. to seek for a faithful, geometrically computable entanglement measure for mixed states. Some progress has been made, such as the entanglement of purification \cite{Terhal_2002,Takayanagi:2017knl,Umemoto_2018,Umemoto_2019}, and partial entanglement entropy \cite{Wen_2020,Lin:2023orb,Lin:2021hqs}. However, {neither of these measures excludes classical correlations} apart from the quantum entanglement. Moreover, the challenges become even more formidable when we attempt to study multipartite entanglement structures \cite{Ju:2023tvo}, as these structures are significantly richer than their bipartite counterparts \cite{bengtsson2016brief}.

In this paper, we utilize the squashed entanglement \cite{Christandl_2004,Yang_2009,Avis_2008} as the measure of mixed state entanglement between two Rindler-convex subregions in gravitational spacetime. Squashed entanglement is an ideal bipartite quantum entanglement measure \cite{Brand_o_2011}, even though its computation, even in quantum information theory, remains an almost insurmountable challenge. We find that the squashed entanglement between two Rindler convex subregions could vanish when a geometrical condition, which involves the separation of these two subregions by a Rindler-convex region, is satisfied. This condition results in a separable state with no quantum entanglement between the two subregions. The result, namely the separation theorem, is obtained utilizing the generalized Rindler wedge formulated in \cite{Ju:2023bjl,Hubeny_2014,Balasubramanian:2013rqa}.

Furthermore, the generalization of the separation theorem from bipartite to multipartite subregions will also be realized {in} gravitational systems based on the generalized Rindle wedge formalism. As previously mentioned, multipartite entanglement structures offer a {richer and more complicated} landscape compared to their bipartite counterparts. In this work, we study various forms of multipartite quantum state separabilities, including the full $m$-partite separability, the $k$-separability \cite{Horodecki_2009,Hong_2021,Ananth_2015}, as well as the $m$-partite total and partial separabilities \cite{brassard2001multi} {for gravitational subsystems}. Adopting the multipartite squashed entanglement \cite{Yang_2009,Avis_2008} {and the conditional entanglement of multipartite information \cite{Yang_2008}} as our measure, we get various geometrical conditions corresponding to {the vanishing conditions for both measures in each form of the quantum state separability mentioned above}, thus generalizing the separation theorem to multipartite cases.  {Note that all these findings only rely on the fact that Rindler convex subregions are well-defined gravitational subsystems with their entanglement entropy proportional to the surface area and this fact is also supported by the subalgebra subregion duality and the generalized gravitational wedge, besides the generalized Rindler wedge formalism. }

We find an observer interpretation of the multipartite separation theorem in the gravitational system and then introduce a time cutoff {in the gravitational spacetime}, which we propose as a means to probe more explicit structures of entanglement. 
 We have discovered that these nontrivial aspects of entanglement structures in gravitational theory naturally emerge in boundary physics through the utilization of the generalized Rindler wedge subregion duality proposed in \cite{Ju:2023bjl}.

In short, in this work, based on the fact that the entanglement entropy of Rindler convex regions is proportional to the area of the enclosing Rindler convex surfaces, utilizing the squashed entanglement measure, we obtain the bipartite and multipartite entanglement structures of gravitational subsystems, especially the configurations with no quantum entanglement between the subsystems. The rest of the paper is organized as follows. Section 2 provides an overview of the generalized Rindler wedge \cite{Ju:2023bjl}, which serves as one of the basis for the present work. Section 3 is dedicated to establishing the bipartite separation theorem for gravitational subsystems. In Section 4, we show the physical implications of the separation theorem and explore its holographic correspondence. In Section 5, we extend our analysis from bipartite to multipartite entanglement and provide a simple example of three subregions with a time cutoff introduced that probes the more explicit entanglement structure.

\section{Review of {Rindler convex regions as well-defined gravitational subsytems} }

\noindent  {Rindler convex regions are well-defined gravitational subsytems as seen from the subalgebra subregion duality \cite{Leutheusser:2022bgi}, the generalized Rindler wedge formalism \cite{Ju:2023bjl} and the generalized entanglement wedge in \cite{Bousso:2022hlz}.} In this section, we present a concise review of \cite{Ju:2023bjl}, where we studied the gravitational entropy of general observer horizons, as the basis for later discussions. In section 2.1, we demonstrate the Rindler-convexity condition for defining gravitational subsystems as well as the corresponding entanglement entropy. In section 2.2, we briefly discuss the holographic interpretation of the generalized Rindler wedge, which serves as the dual to a boundary time band. Furthermore, we emphasize the significance of the time cutoff of the boundary time band and its introduction into the bulk in section 2.3.

{Note that recently, works by Balasubramanian et al. \cite{Balasubramanian:2023dpj} and Bousso et al. \cite{Bousso:2022hlz} have defined gravitational entanglement wedges for arbitrary gravitational subregions. These studies suggest that the entanglement entropy of the entanglement wedge for a given gravitational subregion is determined by the minimal area of the surface enclosing it. In the case of convex subregions, which we are examining, this minimal area corresponds precisely to the surface area of the convex region. Consequently, our findings in \cite{Ju:2023bjl} can be derived using an approach similar to those outlined in \cite{Balasubramanian:2023dpj} and \cite{Bousso:2022hlz}.} Together with some other supporting facts, including the subalgebra subregion duality, this further confirms the consistency of this new proposal of generalized Rindler wedges \cite{Ju:2023bjl}. 

    \subsection{Rindler convexity condition of the observer horizon}
\noindent In a gravity system, due to the lack of locality, we cannot define gravitational subsystems arbitrarily \cite{Giddings:2018umg, Giddings:2019hjc}. In \cite{Ju:2023bjl}, we proposed a new definition of gravitational subsystems from the spacetime subregion that a set of well-defined accelerating observers could observe. A shape constraint, namely the Rindler convexity condition, is required on the enclosing spacelike surface when defining these subsystems due to the nonintersecting condition of accelerating observer worldlines. We have shown that this Rindler-convexity condition that we proposed is consistent with the constraint in the subregion-subalgebra correspondence on the shape of the subregion in \cite{Leutheusser:2022bgi}.  Our method to define gravitational subsystems provides a physical origin for the consistent definition of the corresponding subalgebra.

In this framework, the properties of the most general horizon for observers were explored, where the term ``observers" denotes a set of detectors distributed across a portion of the Cauchy slice, moving along their individual worldlines without back-reaction to the spacetime, thereby forming a reference frame. These detectors are capable of acceleration and may not be rigid, meaning that the distances between them can vary. However, it is crucial to note that the worldlines of these detectors cannot intersect with each other in order to maintain a well-defined reference frame.


\begin{figure}[H]
    \centering    \includegraphics[width=0.6\textwidth]{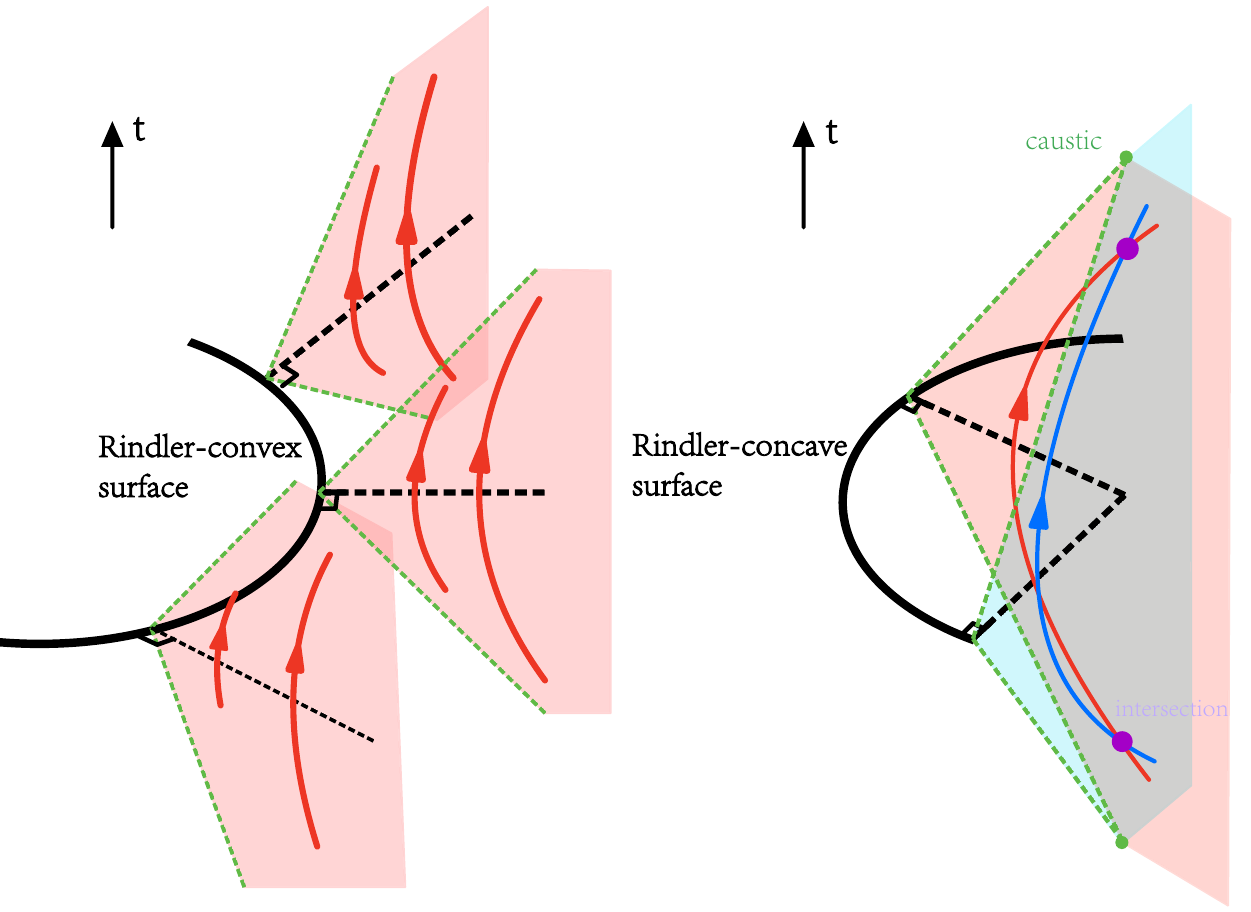} 
    \caption{Accelerating observers associated with Rindler convex and concave surfaces. The worldlines of observers are marked by arrowed curves (hyperbolas). The dashed green lines are {\it normal null geodesics} of the black space-like curves. It can be seen that the normal null geodesics and worldlines of the observers will not intersect when the surface is Rindler-convex (left), while intersect to form caustics (green dots) and the intersection point (purple dots) when the surface is concave (right). Figure copied from \cite{Ju:2023bjl}.} \label{convexconcave}
\end{figure}

Due to the definitions mentioned above, we find that a surface on the Cauchy slice can serve as the bifurcation surface of an observer horizon if and only if it satisfies a global condition named the Rindler-convexity condition, to avoid the intersection of observer worldlines. To determine whether a given spacelike surface is Rindler-convex or not, we have the following two equivalent criteria: the ``normal condition" and the ``tangential condition" \cite{Ju:2023bjl}.

    \textit{\bf The two criteria to determine Rindler-convexity\footnote{ {Note that similar to convexity in flat Euclidean space, Rindler-convex condition describes the orientable co-dimension-2 surface bounding this spatial region, while the region determines the ``inward" side of this surface.}}:} the normal condition states that a surface on a Cauchy slice is `Rindler-convex' if the normal null geodesics outside its boundary both to the future and from the past never intersect to form caustics \footnote{Note that generalized Rindler wedges are different from entanglement wedges. An entanglement wedge is only a generalized Rindler wedge when this entanglement wedge coincides with the causal wedge, with spherical regions in the vacuum state of AdS/CFT serving as an illustrative example. Therefore, entanglement wedges that are not generalized Rindler wedges could encounter caustics, which confirms that they are not generalized Rindler wedges from this normal definition. The main differences  between our work and some previous related concepts \cite{Bousso:2002ju,Gentle:2015cfp,Wen:2018mev,Chen:2022eyi} could be found in \cite{Ju:2023bjl}.}; the {tangential condition for states that}
    the boundary of a region on a Cauchy slice is `Rindler-convex' if any lightsphere externally tangential to its boundary never reaches the inside of the region, where a `lightsphere' is defined to be the intersection of the Cauchy slice with an arbitrary lightcone.

    These two equivalent conditions are derived from the fact that the worldlines of near-horizon observers cannot intersect and that the observer cannot have any causal connection with the region on the other side of the horizon, respectively. {Here we need to emphasize that though these two criteria are not expressed in formulas, they are exactly and explicitly formulated using the mathematical language above. The two conditions are strict in mathematical definitions.} 

    {The normal condition comes from the original definition of Rindler convexity while the tangential definition is more convenient in use. Here we provide a simple proof of the equivalence of these two definitions of Rindler convexity, which is Theorem 2.2 in the following. Before proving this theorem, we first introduce a lemma as follows. }
    {\begin{lemma}
        Given a bifurcation surface $M$ on the Cauchy slice $\Sigma$. One of its normal null geodesics $N_l$ intersects with it at point $P$. The lightsphere on $\Sigma$, generated by a light cone whose vertex is $V$ on $N_l$, must be tangential to $M$ at point $P$, as shown in Figure \ref{Lemma}.
    \end{lemma}
    \begin{figure}[H]
    \centering    \includegraphics[width=0.7\textwidth]{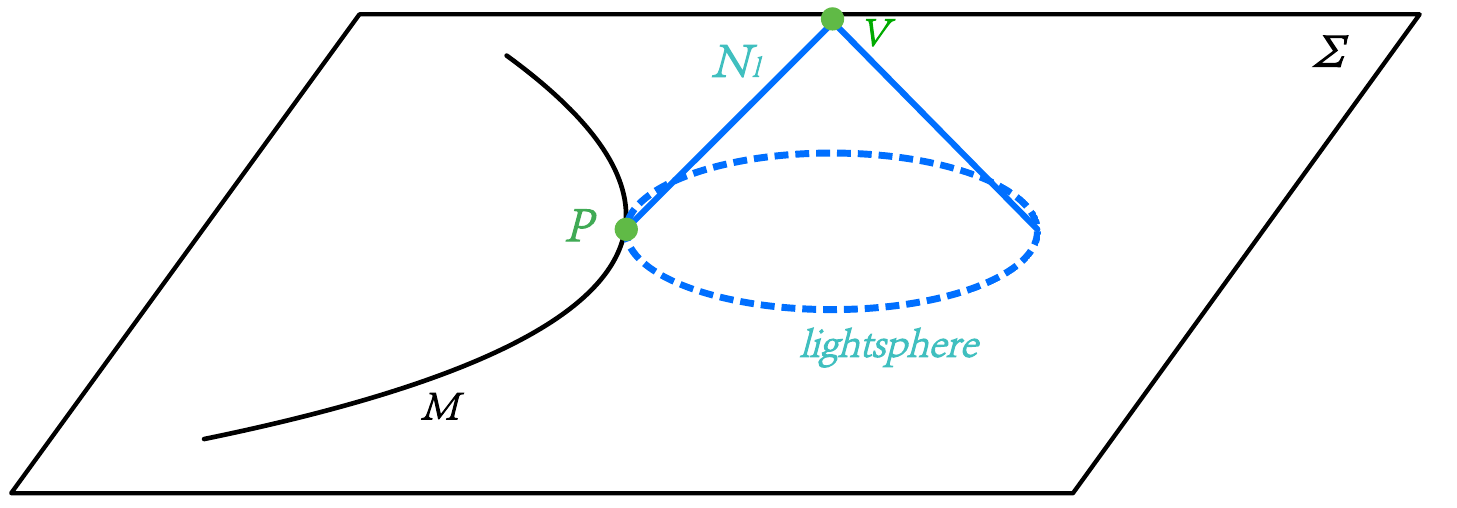} 
    \caption{On the Cauchy slice $\Sigma$, the lightsphere will be tangential to the bifurcation surface at point $P$.} \label{Lemma}
\end{figure}
\begin{proof}
Let $C$ and $\underline{C}$ denote the (future) outgoing and incoming null geodesic congruence normal to $M$, respectively \cite{Aretakis}, (proposition 2.4.2). Then,
\begin{equation}\label{nullgeodesiccongruence}
    \partial \mathcal{J}^{+}(M) \subset C \bigcup \underline{C},
\end{equation}
 {where the left hand side denote the boundary of the causal future of $M$.}
The point inside the lightsphere is causally connected with point $V$. If this lightsphere is not tangential to $M$, i.e., the lightsphere crosses $M$ at $P$, then there exists a point $P' \in M$, a neighbor of $P$ inside this lightsphere and causally connected with point $V$, which violates equation (\ref{nullgeodesiccongruence}). Q.E.D. 
\end{proof}
}
{\begin{theorem}
    \textit{{ The normal condition and the tangential condition of Rindler convexity \footnote{Note that Rindler-convexity is a covariant concept that does not rely on reference frames.  Moreover, as Weyl transformation leaves null geodesics invariant, Rindler-convexity is also Weyl invariant.} are equivalent.}}
\end{theorem}
\begin{proof}
      {Given a spatial region $H$ on the Cauchy slice $\Sigma$, its co-dimension-1 spatial boundary is denoted as $\partial H$.}
     If the normal condition is violated for $\partial H$, i.e., the normal null geodesics intersect to form a caustic $C$ on the $t=t_c$ Cauchy slice. The lightsphere formed by the intersection of the light cone at $C$ and the $t=0$ Cauchy slice must be tangential to $\partial H$ at two points. Then,  {a larger} lightsphere tangential to $\partial H$ at  {one of those two points} which must reach the inside of the region $H$, and the tangential condition is violated. On the other hand, if the tangential condition is violated, there must exist a lightsphere $L$ tangential to $\partial H$ at two points. As a result, the normal null geodesics emitted from those two points must intersect to form a caustic $C$ where the light cone on it intersects with the $t=0$ Cauchy slice, forming the lightsphere $L$, and the normal condition is violated. To summarize, the normal condition is equivalent to the tangential condition.
\end{proof}
}

    In Figure \ref{convexconcave}, we illustrate the Rindler concave surfaces and Rindler convex surfaces from the (non)intersections of worldlines of accelerating observers. We can see from the figure that when a set of accelerating observers could exist whose worldlines do not intersect, the bifurcation surface of these accelerating observer horizons is Rindler convex. On the other hand, when the accelerating observer worldlines cannot avoid intersecting with each other, the corresponding bifurcation surface of their horizons is Rindler concave. Typical examples of Rindler convex surfaces in gravitational spacetimes include the {bifurcation surface of cosmological horizon in dS spacetime, bifurcation surfaces of causal wedges in AdS spacetime, spherically symmetric surfaces outside the trapped/anti-trapped surface in spherically symmetric spacetime, etc.}
    Explicit geometric constructions of Rindler convex surfaces could be constructed from these two normal and tangential conditions,  examples of which could be found in \cite{Ju:2023bjl}, and we provide explicit geometric constructions in several types of gravitational spacetime in {Appendix A}.

    The Rindler convex condition simplifies in certain special geometries, e.g. in a $T_{\mu\nu}k^\mu k^\nu=C_{\rho\mu\nu\sigma}k^\rho k^\sigma=0$ spacetime (null vacuum \footnote{
    A null vacuum spacetime background allows a surface to be Rindler-convex on both sides, i.e. well-defined accelerating observers could live on either side of the surface. To our knowledge, null vacuum is a very strict condition, and viable examples include the Minkowski spacetime, the AdS, and the dS spacetime.}) on a Cauchy slice with a vanishing extrinsic curvature, Rindler-convexity is equivalent to geodesic convexity \cite{Ju:2023bjl}.  Additionally, when a black hole is present, the gravitational lensing effect ensures that any Rindler-convex region must encompass the event horizon.
    We name the wedge in the spacetime which observers can access the generalized Rindler wedge (GRW).
    
    The perspective of accelerating observers offers a method for consistently defining gravitational systems, which provides a physical origin for the convexity condition in defining the corresponding subalgebra in \cite{Leutheusser:2022bgi}. By employing reference frame transformations, we can exclude the region within the horizon and compute the gravitational entanglement entropy of the GRW using the replica trick. Further details can be found in \cite{Balasubramanian:2013rqa, Ju:2023bjl}. Ultimately, this approach leads to the entanglement entropy of the GRW being proportional to the horizon area
    \begin{equation}
        S = \frac{A}{4G}.
    \end{equation}
    
    {Note that {this procedure to calculate the entanglement entropy follows and generalizes \cite{Balasubramanian:2013rqa} for the ``hole" in flat spacetime, while this procedure is also similar to what \cite{Casini:2011kv} has done in the bulk for calculating the entanglement entropy of a spherical subregion at the boundary. In \cite{Casini:2011kv}, a bulk isometry transforms the RT surface of the boundary spherical region to a horizon of a topological black hole with a hyperbolic horizon, which is in fact a hyperbolic surface in AdS. After the transformation, the entanglement entropy in the original spacetime becomes the thermal entropy in the spacetime with a horizon. Here similarly, we could also have a thermal system after transforming to the observer spacetime.  The difference is that here we do not require the transformation to be an isometry as we do not need to calculate the entanglement entropy from the corresponding thermal entropy. As we already have a well-defined gravitational subsystem, we could directly calculate the entanglement entropy from a gravitational path integral method utilizing the replica trick following \cite{Balasubramanian:2013rqa}. As we will see in the next section, in fact \cite{Casini:2011kv} corresponds to the case that EW=GRW, i.e. the entanglement wedge of the boundary spherical region is a GRW, one special case here. A calculation of the modular Hamiltonian in our system needed in the holographic construction of \cite{Casini:2011kv,Song:2016gtd,Jiang:2017ecm} will be presented in another work \cite{Ju:2025mvz}.}}

Up to this point, an important consequence for introducing Rindler-convexity becomes apparent: it guarantees that the surface area of a subregion is smaller than the region that contains it, provided both are Rindler-convex, as proven under the null energy condition in \cite{Ju:2023bjl}. This property is crucial since the holographic principle informs us that the gravitational degrees of freedom of a region reside on its spatial boundary and a potential issue arises when the region is highly irregular with an exceptionally large surface area. The Rindler-convex condition prevents this problem by precluding the region from being excessively ``rough".

 {In \cite{Balasubramanian:2023dpj,Bousso:2022hlz}, the authors generalized 
the entanglement wedge to the gravitational generalized entanglement wedge, which is the bulk spacetime subregion $EW(a)$ that could be reconstructed from a given bulk subregion $a$, with its bounding surface the minimal surface covering $a$. This gravitational entanglement wedge is also a new method to define gravitational subsystems. We now prove a theorem to show that under reasonable conditions, a Rindler convex region is also a generalized entanglement wedge defined in \cite{Balasubramanian:2023dpj,Bousso:2022hlz}.
\begin{theorem}
\textit{If the null energy condition (NEC) holds, then on a Cauchy slice \(\Sigma\), within the causally connected subset\footnote{Here “causally connected subset” means the set of points \(p\in\Sigma\) for which there exists at least one future- or past-directed null geodesic from \(p\) that intersects a null geodesic from the Rindler-convex region. In AdS or Minkowski this is all of \(\Sigma\); in special spacetimes such as de Sitter or “bag of gold” spacetimes it may be a proper subset. In the former spacetimes, a Rindler convex region is proved to be the spatial part of a generalized entanglement wedge.}, the minimal surface enclosing a Rindler-convex region coincides with its bounding surface.}  
\end{theorem}
\begin{proof}
By Raychaudhuri’s equation for a null geodesic congruence, under the NEC any null congruence with negative expansion will inevitably develop caustics. Since we require the outgoing normal null congruence \(\Theta_1\) of the Rindler-convex surface \(R\) to be non-intersecting, \(\Theta_1\) must always have non-negative expansion.
Now suppose there exists a minimal hypersurface \(\xi\) of zero extrinsic curvature that also encloses \(R\). The ingoing normal null congruence \(\Theta_2\) orthogonal to \(\xi\) must then have non-positive expansion by Raychaudhuri’s equation. Let \(Int\) denote the intersection of \(\Theta_1\) and \(\Theta_2\); \(Int\) is a spacelike, codimension-2 surface. Because \(\Theta_1\) has non-negative expansion and \(\Theta_2\) has non-positive expansion, the areas satisfy
\[
\mathrm{Area}(R)\;\le\;\mathrm{Area}(Int)\;\le\;\mathrm{Area}(\xi).
\]
Hence, among all surfaces enclosing \(R\), its bounding surface \(R\) has the smallest area, as claimed.
\end{proof}
}

    \subsection{Holographic dual of GRW and the boundary  time band}
\noindent In the context of AdS/CFT correspondence, the concept of ``hole-ography" relates differential entropy on the boundary to the non-extremal surface area in the bulk \cite{Balasubramanian:2013lsa,Balasubramanian:2018uus,Engelhardt:2018kzk}. 
The concept of the shape constraint for the hole, which is equivalent to the condition of Rindler-convexity, was also introduced by Hubeny in \cite{Hubeny_2014} in the study of differential entropy. This idea was motivated by the reversibility of the null geodesics shooting from a hole in the bulk and a time strip on the boundary, which, combined with the causality constraint, corresponds to the tangential condition of Rindler-convexity. The information-theoretic interpretation of differential entropy was explicitly discussed in \cite{Czech_2015}. 

However, differential entropy is not a fine-grained entropy that could be a von Neumann entropy of a certain quantum state. In \cite{Ju:2023bjl}, we proposed a corresponding interpretation of the differential entropy as the entanglement entropy of a CFT state with long-range entanglement being ``cut" utilizing the boundary time cutoff induced by bulk GRWs, namely the GRW subregion duality.

The proposed GRW subregion duality suggests that the GRW in the bulk corresponds to a boundary spacetime subregion that results from the intersection of the GRW with the boundary. This duality can be understood as an extension of concepts such as causal holographic information \cite{Hubeny:2012wa,DeClerck:2019mkx}, ``time-band subalgebra-subregion" duality \cite{Leutheusser:2022bgi}, and AdS-Rindler subregion duality \cite{Parikh:2012kg,Sugishita:2022ldv}. However, the challenge in proving this duality lies in mathematically defining the ``time cutoff" on the boundary. In \cite{Balasubramanian:2013lsa,Balasubramanian:2018uus}, it is suggested that the differential entropy  reveals the entanglement between scales \cite{Balasubramanian:2011wt}, specifically the entanglement between UV/IR degrees of freedom on the boundary divided by the time cutoff. However, this interpretation is somewhat ambiguous because it is difficult to explicitly determine which IR degrees of freedom have been cut due to the energy-time uncertainty of the observers on the boundary, which is constrained to a finite time interval. \cite{Ju:2023bjl} proposed an alternative interpretation: the observers can observe {\it all degrees of freedom} on the boundary, but the observer cannot detect the long-range entanglement structure due to the causality constraint imposed by the time strip.

To explain this explicitly, Figure \ref{bcutoff} shows a boundary spacetime subregion with a ``zigzag" shape time cutoff, where the vertical axis is the time direction. This ``zigzag" time cutoff excludes entanglement between regions $A$ and $B$ as their causal diamond is not complete due to the time cutoff, while preserving the entanglement structure within regions $AE$ and $BE$ \cite{Ju:2023bjl}. This exclusion corresponds to the {vanishing of the} conditional mutual information $I(A;B|E)$. We can interpret this process as counting the ``out legs" after the time cutoff, where the ``legs" illustrate the entanglement between infinitesimal subregions, as shown on the right side of figure \ref{bcutoff}. Therefore, we can get the residual entropy \cite{Hubeny_2014} associated with this ``zigzag" spacetime subregion as:
\begin{equation}\label{resid}
S^{res}_{ABE}=S(ABE)+I(A;B|E)=S(AE)+S(BE)-S(E).
\end{equation}

{As an example, for a 3-partite GHZ state $\ket{GHZ}_{ABE}=\frac{1}{\sqrt{2}}(\ket{0_A0_B0_E}+\ket{1_A1_B1_E})$, its reduced density matrix is $\rho_{AB}=\frac12(\ket{0_{A}0_{B}}\bra{0_A0_B}+\ket{1_{A}1_{B}}\bra{1_A1_B})$, the conditional mutual information $I(A;B|E)=S_{AE}+S_{BE}-S_E-S_{ABE}=\log2$. We can construct another state $\tilde\rho_{ABE}=\frac12(\ket{0_{A}0_{B}0_{E}}\bra{0_A0_B0_{E}}+\ket{1_{A}1_{B}1_{E}}\bra{1_A1_B1_{E}})$. It has the same reduced density matrices ($\rho_{BE}$ and $\rho_{BE}$) as $\rho_{ABE}$, while the conditional mutual information of this $\tilde\rho_{ABE}$ state vanishes. In this sense, ``cutting" the entanglement between $A$ and $B$ in $\rho_{ABE}$ results in the $\tilde\rho_{ABE}$ state, with its von Neumann entropy given by equation (\ref{resid}).}

{In the same vein,} considering the entirety of $S^{res}_{ABE}$ and introducing additional regions $C, D, \ldots$, we can extend the counting process along the spatial direction. If we take regions $A$, $B$, $C, \ldots$ to be sufficiently small such that we can take the continuous limit, this formula will coincide with the formula for the differential entropy \cite{Balasubramanian:2013lsa}.

\begin{figure}[H]
    \centering    \includegraphics[width=0.9\textwidth]{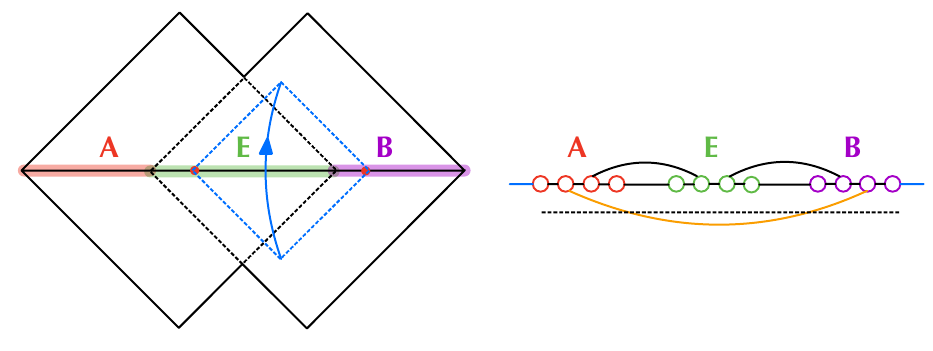} 
    \caption{Left: A boundary `zigzag' shape spacetime subregion with a time cutoff. The horizontal axis is the spatial direction and the vertical axis is the time direction. A boundary observer, following the blue curve as their worldline, has the capacity to observe the entanglement structure within a causal diamond outlined by dashed blue lines. This observer can freely exist within the zigzag region delimited by the time cutoff. Consequently, they can observe any entanglement structure within regions $AE$ or $BE$. However, the entanglement structure between regions $A$ and $B$ remains unobservable due to the time cutoff. Right: The illustration showcases the entanglement structure. Solid lines and curves represent the `entangling pairs,' and the number of blue lines indicates the entanglement entropy $S_{ABE}$. The black lines and curves depict the observable entanglement structure, while the orange curve signifies the entanglement that becomes unobservable due to the time cutoff (represented by the dashed line). Figure taken from \cite{Ju:2023bjl}.} \label{bcutoff}
\end{figure}

With this understanding, we can also give a large class of causal holographic information an information-theoretic interpretation as follows. For the causal wedge in the bulk which is exactly a GRW, the intersection of it with the boundary corresponds to the causal diamond of the boundary subregion. Consider an oval boundary subregion as an illustration. Its causal diamond can be perceived as a ``time cutoff" mechanism that preserves entanglement for distances shorter than the minor axis of the oval, while cutting off correlations beyond this range. The causal holographic information \cite{Headrick:2014cta,Hubeny:2012wa,DeClerck:2019mkx} is the number of the remaining `out legs' after this cutoff.

One remaining question is whether the concept of a ``reconstructed" state, where all long-range correlations are ``cut", exists in the context of quantum information theory. In \cite{Ju:2023bjl}, tripartite toy models involving GHZ and W states are discussed. The conclusion is that the former can be ``reconstructed," while the latter cannot. However, it is believed that the process of cutting all long-range correlations can always be achieved in holographic field theory. We substantiate this proposition by introducing a ``space cutoff" in the bulk geometry, which removes the IR region inside the Rindler-convex hole on the AdS-Cauchy slice. Through an analysis of the entanglement wedge shapes, we argue that this bulk geometry precisely corresponds to the reconstructed state on the boundary. 

\subsection{Further consequences and spacetime cutoff in the bulk}

In this subsection, we summarize some further consequences of this GRW physics, especially on the physics of imposing spacetime cutoff in the bulk.

{\it Bulk spacetime cutoff.} After the discussion of time cutoff on the boundary, we extend the concept of time cutoff from the boundary field system to the gravitational system. This provides a novel approach to investigate the entanglement structure in gravitational systems. By introducing a time cutoff in the bulk, we observe that the condition of Rindler-convexity is relaxed, resulting in certain regions that were initially Rindler-concave becoming Rindler-convex because their normal null geodesics intersect in the spacetime region now removed by the time cutoff. Significantly, we demonstrate that this bulk time cutoff also serves to ``cut" the long-range entanglement in the gravitational spacetime, analogous to the effect of the boundary time cutoff on the boundary field theory\footnote{This analogy to the effect does not imply that this time cutoff in the bulk has an exact holographic correspondence.}.

{\it Consequences of bulk spacetime cutoff.} According to our observer interpretation, the entanglement structure between regions $A$ and $B$ is eliminated due to the inability of observers in these regions to communicate within a finite time interval {after imposing the time cutoff in the gravitational spacetime}. This effectively renders the entanglement structure between $A$ and $B$ unobservable. From a geometric perspective, this implies that the normal null geodesics emitted from the surfaces of $A$ and $B$ never intersect, either in the future or in the past. Consequently, the combination of regions $A$ and $B$ becomes Rindler-convex. Typically, Rindler-convex regions are topologically trivial because the normal null geodesics emitted from disjoint subregions $A$ and $B$ must intersect, violating the convexity condition. However, the introduction of a time cut-off relaxes the Rindler convexity condition, allowing the disconnected region $AB$ to become Rindler convex with an entanglement entropy of $S_{AB} = \text{Area}(A+B) = S_A + S_B$. As a result, the mutual information between regions $A$ and $B$ vanishes, indicating that the entanglement between them is effectively ``cut off". 

The observation that a time cutoff can ``cut" long-range entanglement can be utilized to probe the structure of bipartite or multipartite entanglement in a gravitational system. We will explore this topic in more detail in the following sections.

{\it Holographic observer concordance.} The concluding aspect of \cite{Ju:2023bjl} that we would like to introduce is the concept of holographic observer concordance. Briefly, this refers to the fact that bulk subsystems partitioned by bulk observers are dual to boundary subsystems that are partitioned by boundary observers induced from these bulk observers at the boundary, separately. This statement holds true even when a time/space cutoff is introduced in the bulk, leading to the concept of holographic observer concordance, asserting that the partitioning of degrees of freedom by observers in the bulk and the boundary is in concordance through holography. 
The utility of this holographic observer concordance concept will be of much use in Section 4, where we aim to establish the holographic correspondence of the separation theorem on the boundary, building on its observer interpretation in the bulk.

\section{Bipartite entanglement for gravitational subsystems}
\noindent {Studying the entanglement entropy of more general gravitational subregions} helps us understand the shape dependence of entanglement entropy of more general gravitational subsystems. To discern more details about the entanglement structure, the next step should be to analyze the quantum entanglement between different subregions of the gravitational system {for these general gravitational subsystems}, in order to find out the more explicit quantum entanglement structure. Therefore, an immediate question arises: what measure should we use to characterize the quantum entanglement between gravitational spatial subregions? 

Considering the whole Cauchy slice in the gravitational spacetime as a pure state, the union of two or many spatial subregions is in general a mixed state. As von Neumann entropy is not an adequate measure for quantum entanglement in mixed states, we need to use alternative quantities to assess the entanglement between spatial subregions. The analysis of quantum entanglement in mixed states has been a topic of study for a long time with many different measures proposed to quantify the quantum entanglement. In this work, we utilize squashed entanglement, which is a faithful \cite{Brand_o_2011} but challenging measure to compute. Being a faithful measure, the subsystems are separable if and only if the squashed entanglement is zero, and we aim for a way to determine when it equals zero for our gravitational subsystems. We find a sufficient condition under which the squashed entanglement vanishes so that the quantum entanglement disappears and provides a correspondence between the separability of quantum states in the gravitational subregions and the geometric structure of the subregions. We will first concentrate on the bipartite case, and then we extend our results also to multipartite entanglement in the subsequent section.

In this section, we will first review the definition and properties of squashed entanglement in quantum information theory and then derive the sufficient condition for zero squashed entanglement for gravitational subregions utilizing the Rindler convexity condition. We will explain this result with observer physics and compare the results with previous holographic results revealing their shortcomings in the calculations. 

\subsection{Squashed entanglement in quantum information theory}
\noindent
There are many entanglement measures for mixed states, such as entanglement of purification \cite{Terhal_2002}, distillable entanglement, entanglement of formation \cite{Bennett_1996pur,Bennett_1996}, relative entropy of entanglement \cite{Vedral_1997}, squashed entanglement \cite{Christandl_2004}, etc. Here, we utilize the squashed entanglement as the quantity to measure the bipartite quantum entanglement of the gravitational degrees of freedom within the subregions. This choice is due to the outstanding properties that squashed entanglement exhibits, which we will mention below. 

Squashed entanglement is defined as one-half of the infimum of the conditional mutual information (CMI):
\begin{equation}
    E_{s q}\left(\rho_{A; B}\right)=\inf \left\{\frac{1}{2} I(A ; B \mid E)_\rho: \rho_{A B E} \text { is an extension of } \rho_{A B}\right\},
\end{equation}
where $I(A ; B \mid E)$ is:
\begin{equation} \label{CMI}
    I(A ; B \mid E):=S(A E)+S(E B)-S(E)-S(A E B).
\end{equation}
Note that the coefficient $\frac12$ is only a convention with which the squashed entanglement will coincide with the von Neumann entropy when $\rho_{AB}$ is a pure state. 

Among various measures of entanglement in mixed states, squashed entanglement satisfies most properties proposed as useful for an entanglement measure, such as non-negativity,  monogamy \cite{Koashi_2004}, asymptotic continuity \cite{Alicki_2004}, convexity, and LOCC monotonicity \cite{Christandl_2004}. Most importantly for our subsequent argument, squashed entanglement is a \textbf{faithful} measure \cite{Brand_o_2011}, meaning that squashed entanglement will vanish \textbf{if and only if} the mixed state $\rho_{AB}$ is separable. A simple explanation of this faithful condition is as follows. It was first proven in \cite{Hayden_2004} (theorem 6) that a state $\rho_{A B E}$ on $\mathcal{H}_A \otimes \mathcal{H}_B \otimes \mathcal{H}_E$ saturates strong subadditivity with equality if and only if there is a decomposition of the system $E$ into a direct sum of tensor products 
\begin{equation}
\mathcal{H}_E=\bigoplus_j \mathcal{H}_{e_j^L} \otimes \mathcal{H}_{e_j^R}
\end{equation} such that 
\begin{equation}\label{bimarkov}
    \rho_{A B E}=\bigoplus_j p_j \rho_{A e_j^L} \otimes \rho_{e_j^R B},
\end{equation} with states $\rho_{A e_j^L}$ in $\mathcal{H}_A \otimes \mathcal{H}_{e_j^L}$ and $\rho_{e_j^R B}$ in $\mathcal{H}_{e_j^R} \otimes \mathcal{H}_B$, and classical probabilities $p_j$. Take partial trace on the subsystem $E$ in equation (\ref{bimarkov}), and we could get
\begin{equation}
    \rho_{A B}=\sum_j p_j \rho_{A_j} \otimes \rho_{B_j}.
\end{equation}
\ie{} the $AB$ reduction of state $\rho_{ABE}$ is separable indicating vanishing quantum entanglement between A and B. 

Squashed entanglement, serving as a lower bound on the entanglement of formation and an upper bound on distillable entanglement, possesses nearly all the properties we desire for measuring quantum entanglement \cite{Headrick:2019eth}, excluding classical correlations \cite{Vedral_2003}. However, its calculation is challenging because finding the infimum of CMI among all possible extensions of $\rho_{ABE}$ is complicated, as the dimension of $\mathcal{H}_E$ might be arbitrarily large. Nonetheless, in gravitational systems, we have managed to determine the subsystem $E$ minimizing the CMI to be zero in a broad range of situations, and therefore we provide a sufficient condition under which the squashed entanglement vanishes, which we show in the next subsection.

\subsection{Squashed entanglement for gravitational subsystems}
\noindent

As reviewed in Sec.2, in \cite{Ju:2023bjl} we have calculated the gravitational entanglement entropy for spatial subregions satisfying a so-called Rindler convexity condition. 
This global Rindler convexity condition is important as it reflects the nonlocal property of gravity. We now use the measure of squashed entanglement to detect the entanglement structure of gravitational spacetime, i.e. calculating the squashed entanglement of these Rindler convex gravitational subsystems, which will reveal more interesting consequences of this condition. 
Our fundamental observation is derived from the fact that strong subadditivity is always saturated for Rindler convex regions $AE$, $EB$, and $ABE$
\begin{equation}\label{ASSA}
        Area(AE)+Area(EB)= Area(AEB)+Area(E).
\end{equation}
\begin{equation}\label{SSA}
        S(\rho_{AE})+S(\rho_{EB})= S(\rho_{AEB})+S(\rho_{E}).
\end{equation}
From the definition of (\ref{CMI}), this is a clear indication that the conditional mutual information between subsystems $A$ and $B$ vanishes, given the subsystem $E$ as a condition, i.e. the CMI $I(A ; B \mid E)=0$. 
As the squashed entanglement is the infimum of CMI $I(A ; B \mid E)$ and it is non-negative, we could see that the squashed entanglement between $A$ and $B$ vanishes, i.e. \begin{equation} E_{sq}(A;B)=0.\end{equation}
In other words, there is no quantum entanglement between $A$ and $B$, or equivalently, $A-E-B$ forms a quantum Markov chain in that order \cite{Hayden_2004}. 

This result strongly relies on the fact that subregions $AE$, $EB$, $E$, and $ABE$ all have to be Rindler convex regions. As $E$ is always a Rindler convex region when $AE$ and $EB$ are both Rindler convex regions, here we only need to specify the condition that all $AE$, $EB$, and $ABE$ have to be Rindler convex regions. Thus the conclusion should be that as long as subregions $AE$, $EB$, and $ABE$ are Rindler convex regions, the quantum entanglement between $A$ and $B$ vanishes. 

\subsubsection{A subtlety when $A$ or $B$ is not Rindler convex}
There is a subtle point here. For the strong subadditivity condition in (\ref{SSA}) to hold, subregions of $A$ and $B$ may not need to be Rindler convex as these two subsystems do not appear in the formula. Since regions $A$ and $B$ may be Rindler concave, we may not be able to construct well-defined accelerating observers associated with $A$ or $B$, and it is unclear whether the gravitational degrees of freedom inside a Rindler-concave region could be well-defined or not due to possible nonlocal distributions of gravitational degrees of freedom. For this reason, we denote the remaining Hilbert spaces of $\mathcal{H}_{AE}$ and $\mathcal{H}_{BE}$ after tracing out $\mathcal{H}_{E}$ as $\mathcal{H}_A$ and $\mathcal{H}_B$ respectively. Note that in general, $\mathcal{H}_A$ and $\mathcal{H}_B$ defined in this way are not the Hilbert spaces of the degrees of freedom inside concave subregions $A$ and $B$. For this to be precisely defined, we have to make the following reasonable assumption: if $M$ and $N$ are both Rindler-convex regions and $M\supset N$, we have
\begin{equation}\label{contain}
\mathcal{H}_M=\mathcal{H}_N\otimes \mathcal{H}_{M/N}\quad\quad \text{if} \quad M\supset N.
\end{equation}
This assumption can be understood as follows. If $M\supset N$, the Rindler observers of $N$ have access to an extra region ($M-N$) than the Rindler observers of $M$. Consequently, we anticipate that the ignorance of $M$'s observer will be larger than that of $N$'s observer so that we have (\ref{contain}). Once again, we adopt the conservative viewpoint and do not interpret $\mathcal{H}_{M/N}$ as the Hilbert space of the concave region $M-N$.

Therefore, even though we have demonstrated that $E_{sq}(A;B)=0$, interpreting the physical implications of this result becomes challenging when either $A$ or $B$ is Rindler-concave. As a result, when $A$ or $B$ is Rindler concave, we can alternatively analyze the entanglement structure between two Rindler-convex subregions $a\subset A$ and $b\subset B$ instead {, as shown in Figure \ref{separate}}.
As strong subadditivity is saturated, the whole state $\rho_{ABE}$ can be factorized as in (\ref{bimarkov}). Now, instead of taking the partial trace of subsystem E, we take the partial trace of $\mathcal{H}_{AE/a}$ and $\mathcal{H}_{BE/b}$ with \begin{equation}\mathcal{H}_{AE}=\mathcal{H}_a\otimes \mathcal{H}_{AE/a}, ~~\mathcal{H}_{BE}=\mathcal{H}_b\otimes \mathcal{H}_{BE/b}.\end{equation}

With partial trace $\Tr_{AE/a}$ acting trivially on $\mathcal{H}_{b}$ and $\Tr_{BE/b}$ acting trivially on $\mathcal{H}_{a}$, finally we can get:
\begin{equation}
    \rho_{a b}=\sum_j p_j \rho_{a_j} \otimes \rho_{b_j},
\end{equation}
\ie{} the state is separable (no quantum entanglement exist) for Rindler-convex subregions $a$ and $b$. 
It is important to note that while the condition $E$ ensures that $I(A;B|E)$ vanishes, in general, $I(a;b|E)$ is not zero. Moreover, for general Rindler convex shapes 
 of $a$ and $b$, it is difficult to find a Rindler-convex region $e$ that makes $I(a;b|e)$ vanish. It appears to contradict the fact that $\rho_{ab}$ is a separable state with zero squashed entanglement. However, in the case of squashed entanglement, the choice of the extension $e$ and $\rho_{abe}$ is arbitrary. $e$ does not have to correspond to a state in a Rindler-convex subregion of the whole state, and it does not even have to be a reduced state derived from the entire pure state. This is because the quantum entanglement between $a$ and $b$ is determined solely by $\rho_{ab}$, regardless of which state $\rho_{ab}$ is reduced from. Although finding the optimal extension $\rho_{abe}$ with vanishing $I(a;b|e)$ is challenging, based on the previous conclusion that $\rho_{ab}$ is a separable state, we know that such an extension must exist and therefore the conclusion that $\rho_{ab}$ is a separable state with zero squashed entanglement is valid.

\subsubsection{Geometric conditions for Rindler convex regions $a$ and $b$ to have no quantum entanglement}

In the above, we have shown that $\rho_{ab}$ is separable {for Rindler-convex regions $a$ and $b$} if  $a\subset A$ and $b\subset B$, given $AE$, $BE$, and $ABE$ all satisfing the Rindler-convexity condition. Now, we will analyze the {\it geometric conditions under which Rindler-convex subregions $a$ and $b$ can be separable having no quantum entanglement}, without specifying the exact regions $A$ and $B$. In other words, we want to determine the condition under which we can construct regions $A \supset a$, $B \supset b$ and $E$ given $a$ and $b$, such that $AE$, $BE$, $E$, and $ABE$ are all Rindler-convex, ensuring the saturation of strong subadditivity.

Geometrically, the answer to this question is that this is only possible if we can construct a convex region $E$ that `\textbf{separates}' regions $a$ and $b$. This leads to the following geometric condition, named the  ``separation theorem". 

\textit{Separation Theorem:} there is no gravitational quantum entanglement between the degrees of freedom inside two Rindler convex spatial subregions $a$ and $b$ on a Cauchy slice, i.e. the whole state is separable for $a$ and $b$, if there exists a Rindler-convex region $E$ that \textbf{separates} $a$ and $b$.

Here, the question is then what ``$E$ separates $a$ and $b$" means. The precise definition of `{separate}' is the follows. There are two equivalent, covariant, and {\it mathematically strict} definitions of the term `{separate}' here.

\textit{Normal Definition of `separate'.}
On a Cauchy slice, subregions $a$ and $b$ are `separated' by $E$ if and only if: any normal null geodesic emitted from $a$ and $b$ does not intersect before they intersect with the normal null geodesic emitted from $E$.

\textit{Tangential Definition of `separate'.} Any lightsphere externally tangential to both $a$ and $b$ must pass through $E$. It is worth noting that in the case when Rindler-convexity is equivalent to geodesic convexity, the term `separate' implies that any geodesic which passes through both regions $a$ and $b$ must also pass through the region $E$ (see Figure \ref{separate}).


\begin{proof}
    {\textit{The equivalence of the Normal Definition and the Tangential Definition.}
If the normal definition is violated, we have two normal null geodesics emitted from $a$ and $b$ that intersect at $P$ before they intersect with the normal null geodesic emitted from $E$. Then the lightsphere whose light cone vertex is $P$ will be tangential to both $a$ and $b$ (Lemma 2.1), but does not pass through $E$. This violates the tangential definition. On the other hand, if the tangential definition of ``separation" is violated, \ie{}, there exists a lightsphere tangential to both $a$ and $b$ that does not pass through $E$, then the vertex of the light cone (P) of this lightsphere will be the spacetime point where two normal null geodesics emitted from $a$ and $b$ intersect. As those two null geodesics intersect before they intersect with the normal null geodesic emitted from $E$, this violates the normal definition. 
}
\end{proof}

\begin{figure}[H]
    \centering    \includegraphics[width=0.64\textwidth]{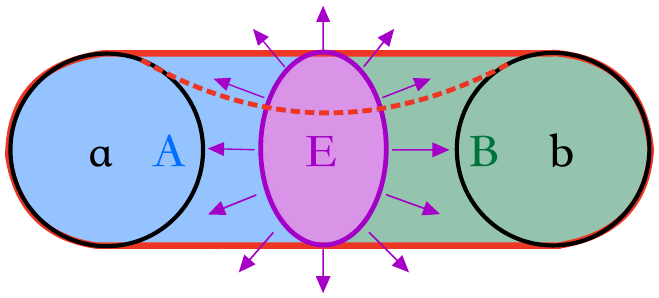} 
    \caption{The spherical regions $a$ and $b$, bounded by the black circle, are separated by the purple-shaded region $E$. The Rindler-convex hull of $ab$ is bounded by the red curve, with $E$ dividing it into the blue-shaded region $A$ and the green-shaded region $B$. The red dashed curve passing through region $E$ represents a lightsphere that is externally tangential to both $a$ and $b$, in accordance with the tangential definition of ``separate". The Rindler observers of the region $E$ reside in regions $A$ and $B$ at this moment are causally disconnected from each other in both the past and future. Regions $AE$, $BE$, $E$, $ABE$, $a$, and $b$ are all Rindler-convex.} \label{separate}
\end{figure}

As shown in Figure \ref{separate}, the definition of `separate' ensures that {I}. region $E$ does not intersect with $a$ or $b$, and {II}. there exist Rindler-convex regions $ABE$, $AE$, $BE$, and $E$, which guarantees the saturation of strong subadditivity (\ref{SSA}). In summary, the separation theorem can be encapsulated in the following slogan.
\begin{equation}\label{slogan}
    \textit{Geometrically separable} \Longrightarrow \textit{Quantum State separable.}
\end{equation}

The contrapositive version of this theorem (the arrow in the opposite direction) would state that if two {Rindler convex} regions cannot be separated by a convex region, then they must have non-vanishing quantum entanglement. Proving this opposite side is indeed difficult because when we cannot find the extension $\rho_{ABE}$ with $I(A;B|E)=0$ on the gravity side, this does not mean that no such a state $E$ could exist in the 
more general range of states in quantum information theory. The challenge arises from the difficulty in calculating $E_{sq}$ in quantum information theory as the dimension of $\mathcal{H}_E$ could be arbitrarily large.

Is this separation theorem too strict? One can conclude that the quantum state on any two non-adjacent regions is separable, i.e. the convex region $E$ is no longer needed, utilizing the same logic based on the gravity theory with any region, including the concave one, having the entanglement entropy proportional to its surface area. However, as we will discuss below when a black hole exists, non-adjacent regions are believed to be entangled due to the intuition from its holographic dual field theory side. This conflicts with the result that any non-adjacent regions are separable, ruling out the possibility of the separation theorem being too strict.

To summarize, though it remains a conjecture whether the contrapositive version of the separation theorem holds, the separation theorem is still believed to serve as a non-trivial effect in the semiclassical limit, which establishes a connection between the geometrical structure and the separability of the corresponding quantum state. 

We summarize the physical findings in this section. In gravitational spacetime, two subsystems of Rindler convex subregions are separable with no quantum entanglement when there is another Rindler convex subregion that separates these two subregions. The explicit mathematical definition of `separate' could be found in section 3.2.2.

\section{Physical interpretations and holographic squashed entanglement}
\noindent

The main message from the previous section is the separation theorem, which states that two Rindler convex subregions in gravity are separable when these two subsystems are separated by a third Rindler convex subregion. In this section, we analyze the physical interpretations of this fact. Compared with the entanglement entropy of a single subregion, the separation theorem provides a more elaborate description of gravitational entanglement structure. This is especially evident when a time cutoff is introduced, as it allows for quantifying the location of entanglement.

\subsection{The location of entanglement}
\noindent
We consider a simple case when $T_{\mu\nu}k^\mu k^\nu=C_{\rho\mu\nu\sigma}k^\rho k^\sigma=0$ (null vacuum). In this case, an infinitely large lightsphere could extend all the way to infinity without converging or diverging. As a result, there exist spacelike surfaces which are convex on both sides. As we proved in \cite{Ju:2023bjl}, if the Cauchy slice has zero extrinsic curvature, these surfaces are planes. Therefore, in a null vacuum, a plane could be an infinitely thin limit of a Rindler convex region, {\it i.e.} Rindler-convex regions can be arbitrarily `thin'. This simplifies the separation of two regions by definition, and two nonadjacent regions can be separated if they can be divided by a plane\footnote{It is worth noting that, when the spatial dimension is greater than two, two regions that can be separated by a convex region are not necessarily separable by a plane. A counterexample is four regions, each near a surface of a tetrahedron; any two of them cannot be separated from the other two by a plane, but they can be separated by the tetrahedron.}. This leads to the conclusion that any two nonadjacent Rindler-convex regions can certainly be separated by a very thin Rindler-convex region, thus any two non-adjacent Rindler-convex regions in a null vacuum as defined above are not entangled.

Imagine two adjacent Rindler-convex regions $a$ and $b$ in a null vacuum, with quantum entanglement between them \cite{Laflamme:1987ec,Czech:2012be,Bousso:2012mh}. If we `remove' a thin interval from $a$ at their interface so that they can be separated by a plane, the quantum entanglement between the new smaller subregion $a$ and $b$ vanishes completely. A plausible physical interpretation of this could be that all entanglement is located near the interface of the subregions. As we have entanglement entropy proportional to its surface area, we would expect a uniform distribution of `gravitational entanglement' near the surface of a region. 

On the other hand, if the spacetime background is not a null vacuum, the Rindler-convex region cannot be arbitrarily thin. In this case, there might exist non-zero gravitational quantum entanglement between two distant Rindler-convex regions. Mathematically, this situation might result from the existence of matter fields or black holes\footnote{Especially when a black hole exists, every Rindler-convex surface must wrap the event horizon. That makes the Rindler-convex condition so strict that only few regions could be geometrically separable.} that deform the vacuum spacetime geometry. We can then conclude that gravitational quantum `entangling pairs' might not be located solely near their interface, when matter or black holes are present. 

\subsection{Introducing time cutoff}
\noindent
We have previously argued that introducing a time cutoff into the gravitational system `cuts out' the long-distance correlation by relaxing the Rindler-convexity condition in gravitational systems \cite{Ju:2023bjl}. {This decrease in correlation includes both quantum entanglement and classical correlation between the two regions. In this context, we are specifically interested in the former and this decrease in quantum entanglement is reflected in two aspects.} With the introduction of a time cutoff: I. the Rindler-convex hull of $ab$ decreases in size, and II. region $E$ could become more `wiggly'. Both of these aspects stem from the relaxation of the Rindler-convexity condition due to the time cutoff, which in turn suggests that separating $a$ and $b$ becomes easier.

{Theoretical analysis allows us to determine when this time cutoff would completely eliminate all quantum entanglement for two subregions {which originally had} quantum entanglement in the general case. The critical thickness of the ``time layer" when the two regions become separable can then naturally reveal the shortest distance of quantum entanglement between them.} 
Generally, as we continue to increase the cutoff to make the remaining `time layer' manifold thinner, the quantum entanglement between $a$ and $b$ would be the first to vanish, followed by the mutual information. In quantum information theory, quantum entanglement contributes to mutual information, meaning that mutual information cannot vanish before the quantum entanglement vanishes. This is evident in the procedure where the `bridge' of the Rindler-convex hull connecting $a$ and $b$, with its cross-section giving an upper bound of the mutual information, becomes thinner, requiring a smaller region $E$ to separate the bridge. Ultimately, when the bridge breaks up, i.e., the Rindler-convex hull is disconnected, and the mutual information vanishes as well. We will provide a simple example to illustrate this in the next section. 

\subsection{An observer concordance explanation}

Following the holographic observer concordance formalism proposed in \cite{Ju:2023bjl}, we also provide an observer interpretation for the understanding of the separation theorem here. 

The separation theorem can be explained through observer physics as follows. {We will show that with the existence of a region $E$ as required in the separation theorem, we could always define a set of physical observers who could not observe $a$ and $b$ at the same time. Thus the existence of this set of well-defined observers, who could not see the quantum entanglement between $a$ and $b$, indicates that $a$ and $b$ should not have any quantum entanglement in this quantum state from the observer concordance proposed in \cite{Ju:2023bjl}. The detail is as follows.}  

With $ab$'s Rindler-convex hull divided by $E$ into two parts (blue and green in Figure \ref{separate}), we can also choose two sets of $E$'s Rindler observers\footnote{When we say $E$'s Rindler observers, we refer to the observers who treat region $E$'s boundary as their horizon.}, with each set of observers constrained in either one part of $ab$'s Rindler-convex hull on the Cauchy slice (blue or green). 
Each set of observers could observe all the entanglement structures inside regions $a$ or $b$, respectively, as these observers can persist for an arbitrarily long time to observe inside regions $a\subset A$ and $b\subset B$.
Due to the normal definition of `separate', the two sets of observers inside regions $A$ and $B$ never did and never will have any causal connection with each other because each set is behind the Rindler-horizon of the other. From these observers' perspectives, gravitational degrees of freedom inside regions $a$ and $b$ did not and will not have any chance to entangle with each other. Furthermore, these two sets of observers will not even be aware of whether they exist in the same spacetime manifold or not.

\subsection{Separation theorem in holography}

\noindent
One can observe that the separation theorem unveils nontrivial entanglement structures within gravitational theory. Naturally, it is important to find its corresponding boundary theory physics to gain further insights.
The AdS/CFT correspondence demonstrates that the theory of Einstein gravity in anti-de Sitter spacetime, as $G_N\to 0$, is dual to a conformal field theory in the large $N$ limit, residing on the conformal boundary of AdS spacetime. Since our preceding argument is general and not tied to any specific spacetime background, it should also hold in the AdS spacetime. 
Conveniently, subregion-subregion duality and GRW subregion duality establish a correspondence between the density matrices in these two theories.
Below, utilizing these two formalisms, we will show that the separation condition in the bulk naturally gives rise to boundary physics that emphasizes the ``entanglement builds geometry" idea.

Entanglement wedge reconstruction \cite{Czech:2012bh,Wall:2012uf,Headrick:2014cta,Espindola:2018ozt,Saraswat:2020zzf,dong2016reconstruction,Harlow:2018fse,Bousso:2012sj,Leutheusser:2022bgi} states that the quantum state $\rho$ on a boundary subregion is dual to the bulk gravity subsystem within its corresponding entanglement wedge, while the GRW subregion duality associates the bulk GRW to a different $\tilde \rho$ state of the same boundary spatial subregion with certain long-range entanglement removed, as reviewed in Section 2. 
The separation theorem specifies the separability of quantum states in {two} Rindler-convex subregions in the bulk, and though EWs are in general not Rindler convex, however, according to the entanglement wedge reconstruction EWs are well-defined subregions whose degrees of freedom are fully encoded within the subregions, thus we could change the smaller Rindler convex subregion $a$ in the bulk separation theorem to EWs. {Therefore we have a separable quantum state of two quantum states in the two EWs $EW(\mathbf{A})$ and $EW(\mathbf{B})$\footnote{ {Note that we always use bold font (\eg, $\mathbf{A}$) to represent boundary subregion on a Cauchy slice and normal font ($A$) denote the bulk subregion. In addition, as we are always working on a single Cauchy slice, we only consider the spatial part of the entanglement wedge, \ie, $EW(\mathbf{A})$ denotes the intersection of $A$'s entanglement wedge and the Cauchy slice $\Sigma$.}}. An immediate question is what is the corresponding boundary separable quantum state and what this separation theorem implies at the boundary. }

As separability of two boundary subsystems is a very strong statement, crucially, one issue we must clarify is: which state are we actually considering for separability? From the formula of zero squashed entanglement, we trace out the subsystem $E$ in (\ref{bimarkov}) to obtain $\rho_{AB}$ and test the separability of $A$ and $B$. Note that the degrees of freedom of $E$ could be traced out as long as the whole Hilbert space could be factorized into the Hilbert spaces of $E$ and its complement. Due to the nonlocality of gravity, this could not always be done \cite{Donnelly:2017jcd, Giddings:2018umg, Donnelly:2018nbv, Giddings:2019hjc}. However, for GRWs or their compliments, we could do this because a consistent subalgebra \cite{Leutheusser:2022bgi,Ju:2023bjl} for these systems could be well-defined\footnote{The algebra mentioned here is a Type III von Neumann algebra in nature, which leads to an obstacle in factorizing Hilbert spaces into tensor products. Howver, in principle, these subalgebras could be modified into Type II von Neumann algebras as has been done in \cite{Witten:2021unn, Chandrasekaran:2022cip, Chandrasekaran:2022eqq,Witten:2023qsv} recently. }. 
On the gravity side, rigorously speaking, tracing out E means that we are assessing the separability of $A$ and $B$ within the GRW outside $E$, i.e. the reduced state. However, in quantum information theory, as long as the degrees of freedom in $E$ ($\mathcal{H}_E$) are independent of the degrees of freedom in $AB$ ($\mathcal{H}_{AB}$), tracing out $E$ will not modify the entanglement structure between $A$ and $B$ as this is only determined by $\rho_{AB}$. In gravitational theory, the concept that ``observers define consistent subsystems" indeed leads to the fact that degrees of freedom in $E$ can be partitioned from the rest. Thus, on the gravity side, we could conclude that we are studying the separability of $A$ and $B$ in the original state, where we define $\mathcal{H}_{A(B)}$ as $\Tr_E\mathcal{H}_{AE(BE)}$. 

In the dual boundary field theory side, things can become subtle because degrees of freedom in the bulk that are dual to the boundary subregions are non-local. This is evident from the fact that the entanglement wedge does not adhere to an addition rule, i.e., \begin{equation} EW(\mathbf{A}\cup \mathbf{B})\neq EW(\mathbf{A})\cup EW(\mathbf{B}).\end{equation} Consequently, a region $E\subset EW(\mathbf{A}\cup \mathbf{B})$ in the bulk might encode information in $\mathbf{A}\cup \mathbf{B}$ on the boundary, even if $E$ has no intersection with $EW(\mathbf{A})$ or $EW(\mathbf{B})$. Therefore, the bulk region $E$ might be understood as being dual to some IR boundary degrees of freedom \cite{Balasubramanian:2013lsa,Balasubramanian_2012} which carry long-range entanglement structures \cite{Gioia_2022}, and tracing them out will actually modify the entanglement structure between $A$ and $B$ on the asymptotic boundary. On the dual field theory side, we assert that $\tilde \rho$ is a state derived from $\rho$ with the elimination of certain long-range entanglement, corresponding to a GRW. From the calculation in (\ref{bimarkov}), it becomes apparent that we are examining the separability of $\mathbf{A}$ and $\mathbf{B}$ in the $\tilde \rho$ state. {This is because the quantum state of $\mathbf{A}\mathbf{B}$ has changed with some of their long-range entanglement removed and this separability in $\tilde{\rho}$ differs from the separability in the state $\rho$ because some degrees of freedom have already been traced out} \footnote{In the case which $\mathbf{AB}$'s RT surface is disconnected, tracing out $E$ will not modify the entanglement structure between $\mathbf{A}$ and $\mathbf{B}$, and we will end up with the conclusion that the squashed entanglement between $\mathbf{A}$ and $\mathbf{B}$ indeed vanishes as we find a region $E$ to separate their entanglement wedges. However, this is a trivial result, as their mutual information vanishes due to the disconnected RT surface.}.

Consequently, in the dual boundary field theory side, the separation theorem tests the separability of subsystem $\mathbf{a}$ and $\mathbf{b}$ in the $\tilde \rho$ state. As we demand that region $E$ cannot have intersections with $EW(\mathbf{a})$ or $EW(\mathbf{b})$, the GRW of $E$ must include $EW(\mathbf{a})$ and $EW(\mathbf{b})$, making $\rho_\mathbf{a}$ and $\rho_\mathbf{b}$ density matrices reduced from $\tilde \rho$. Thus, the geometrical separation between $EW(\mathbf{a})$ and $EW(\mathbf{b})$ corresponds to the fact that one can construct a GRW in the bulk (dual to a time band on the boundary), eliminating the quantum entanglement between $\mathbf{a}$ and $\mathbf{b}$ while preserving $\rho_\mathbf{a}$ and $\rho_\mathbf{b}$. This statement is summarized in Figure \ref{EsqHolo}. {This is not in contradiction to the result on the gravitational side. We could view the subsystem $\mathbf{A}$ or $\mathbf{B}$ in the field theory side as the original subsystem with small momentum modes removed and this could in turn explain the gravitational nonlocality as storing the long-range entanglement structure between subregions.}

\begin{figure}[H]
    \centering    \includegraphics[width=0.6\textwidth]{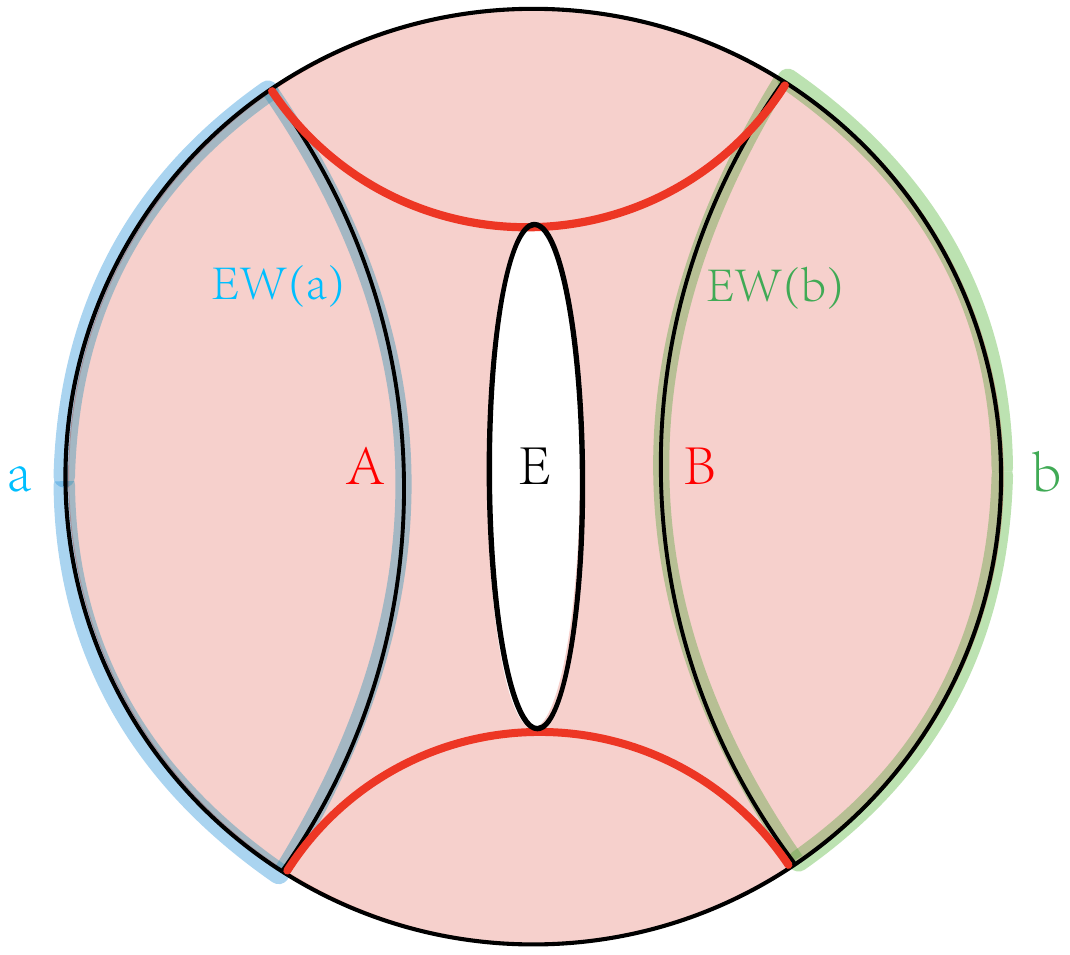} 
    \caption{Separation theorem in holography. Bulk region $E$, $AE$, $BE$ are Rindler-convex regions. $EW(\mathbf{a})\subset A$ and $EW(\mathbf{b})\subset B$ are entanglement wedges of boundary subreions $\mathbf{a}$ and $\mathbf{b}$. Tracing out region $E$ in the bulk corresponds to cutting off the white-shaded region, and the dual state is $\tilde \rho$ which corresponds to the red-shaded region outside region $E$ (red curves are RT surfaces of $\mathbf{ab}$ before $E$ is cut off.). Through evaluation \cite{Ju:2024xcn}, one can find that the conditional mutual information between $\mathbf{a}$ and $\mathbf{b}$ with the region between them being the condition, indeed vanishes in $\tilde \rho$.} \label{EsqHolo}
\end{figure}

\subsubsection{Geometric conditions for boundary separable states in $\tilde{\rho}$}

As we have shown that the bulk separation theorem implies boundary separable states in the $\tilde{\rho}$ state, the next question is the reverse problem: what are the geometric conditions under which boundary subregions have separable dual bulk entanglement wedges. 
Due to the complicate geometrical structure of extremal surfaces in asymptotical hyperbolic space, finding the necessary and sufficient condition would be a difficult task. We now give a sufficient condition, which states that for two boundary subregions $\mathbf{A}$ and $\mathbf{B}$ which can be separated by a spherical shell in the sense that $\mathbf{A}$ and $\mathbf{B}$ are in different sides of the shell, $EW(\mathbf{A})$ and $EW(\mathbf{B})$ are separable subsystems in the bulk.

As shown in Figure \ref{holographicalsep}, this proof is quite intuitive. Since $\mathbf{A}$ and $\mathbf{B}$ on the boundary could be separated by a spherical shell (the blue circle at the boundary), the entire boundary is divided by the shell into two parts: the outside part $\mathbf{M}\supset \mathbf{A}$, and its complement, the inside part $\mathbf{N}\supset \mathbf{B}$. Then, $P=\partial EW(\mathbf{M})=\partial EW(\mathbf{N})$ is the plane (with vanishing extrinsic curvature tensor in the bulk, the blue hemisphere) in the bulk homologous to the spherical {disk} on the boundary. Due to the nesting rule, we have $EW(\mathbf{M})\supset EW(\mathbf{A})$ and $EW(\mathbf{N})\supset EW(\mathbf{B})$. Therefore, {we have proved that $E$ is a Rindler-convex spherical shell that separates $EW(\mathbf{A})$ and $EW(\mathbf{B})$.}

\begin{figure}[H]
    \centering    \includegraphics[width=0.66\textwidth]{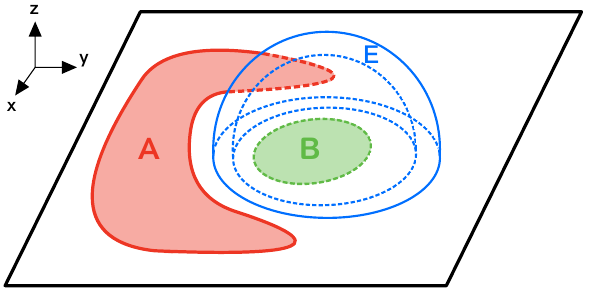} 
    \caption{The separability of $EW(\mathbf{A})$ and $EW(\mathbf{B})$ {by a thin spherical shell $E$ in the bulk} can be illustrated as follows: $\mathbf{A}$ and $\mathbf{B}$ are two boundary subregions shaded in red and green separately, and the blue circle shell on the boundary separates them. One could prove that {the spherical shell $E$} in the bulk homologous to the blue disk must separate $EW(\mathbf{A})$ from $EW(\mathbf{B})$ from the nesting rule.} \label{holographicalsep}
\end{figure}

On the boundary, if regions $\mathbf{A}$ and $\mathbf{B}$ can be ``separated by a spherical shell", it means that we can construct two sets of boundary observers who can observe $\rho_\mathbf{A}$ and $\rho_\mathbf{B}$ respectively without being causally connected with each other in a finite time interval. Then, the time cutoff introduced on the boundary could be seen as the boundary of the causal domain of these observers. {As shown in Figrue 3,} one can observe that this time cutoff indeed ``cuts" all the entanglement between $\mathbf{A}$ and $\mathbf{B}$, which is why we assert that the separation theorem naturally emerges in the context of boundary physics. Observer interpretation also inspires us to understand the theorem above from the perspective of ``holographic observer concordance" framework that 
\textit{the degrees of freedom on the boundary that can be separated by two causally disconnected boundary observers must be dual to the two entanglement wedges that can be separated by two sets of observers in the bulk.}

There are a few more noteworthy points to mention. Firstly, the condition for $EW(\mathbf{A})$ and $EW(\mathbf{B})$ to be separable, as described above, is not a necessary condition in higher-dimensional cases such as $AdS_{d+1}/CFT_d$ with $d>2$. This is because two regions that can be separated by a Rindler-convex region do not necessarily need to be separated by a plane in high dimensional space. Therefore, sometimes it becomes imperative to analyze the specific geometric structure of $EW(\mathbf{A})$ and $EW(\mathbf{B})$ in order to determine whether it is possible to construct a $\tilde\rho$ state where $\mathbf{A}$ and $\mathbf{B}$ are separable in holographic theory.
Secondly, when designing the time cutoff on the boundary to eliminate the entanglement between $\mathbf{A}$ and $\mathbf{B}$, it is important to note that {as the intersection of the edge of a GRW and the asymptotic boundary, the shape of the time cutoff surface cannot be arbitrary}\footnote{{It has been demonstrated in \cite{Hubeny_2014} that the condition of the so-called ``strip wedge" coinciding with the ``rim wedge" restricts the shape of the time band. We will not go through details in this paper.}}.
Thirdly, it is worth noting that the separation theorem holds in general holographic theories, not limited to the AdS vacuum where GRW subregion duality applies. This insight opens up possibilities for extending GRW subregion duality to non-vacuum scenarios and sheds light on the location of entanglement structures. Specifically speaking, in general holographic theories, long-range entanglement structures between regions on the boundary at certain distances are located within the convex IR region which separates their entanglement wedges within the bulk.

\section{Gravitational multipartite entanglement}
\noindent
In this section, we generalize the separable property of the two-partite case to the multipartite case. We first provide an overview of three types of separability of multipartite systems commonly studied in quantum information theory. By employing the measure of multipartite squashed entanglement {and the conditional entanglement of multipartite information}, we are able to identify the distinct geometric configurations of gravitational subregions that correspond to these different types of separability in multipartite states.

Generalizing bipartite entanglement into multipartite entanglement is not a trivial task because the multipartite entanglement structure is much richer than the bipartite case. Although we have already proved that the state of any two non-adjacent Rindler-convex regions is separable in null vacuum spacetimes, this result does not necessarily hold in the multipartite case, i.e. non-adjacent multi-Rindler-convex regions might not be separable as will be demonstrated later. In the first two subsections, we will focus on tripartite entanglement by analyzing three different classes of geometrical tripartite subregions and then generalize our findings to other multipartite cases in the rest of this section. We aim to provide an elucidation that ``geometrically separable" corresponds to ``quantum separable states" for multipartite states in both gravitational systems and the holographic dual field theories. 

\subsection{Separability for mutipartite quantum system}
\noindent Before looking for dedicated multipartite entanglement measures, a natural approach was to employ bipartite quantities to study multipartite entanglement \cite{brassard2001multi}. The separability in bipartite quantum systems could also be generalized to multipartite systems. In the context of general multipartite entanglement structures, there exist three classes of separability.


 {\textbf{1. $m$-party total semi-separability} \cite{brassard2001multi}. The $n$-partite state $\rho_{A_1...A_n}$ is divided into $m$ subsystems $B_1...B_m=A_1...A_n$, where $m\leq n$. The state is $m$-party total semi-separable ($m-TSS$) if and only if:
\begin{equation}\label{FSP}
   \rho_{A_1 \ldots A_n}=\sum_{i} p_i \rho_{B_1}^i \otimes \ldots \otimes \rho_{B_m}^i.
\end{equation}
We denote this by $m-TSS(B_{1};B_{2};...;B_{m})$.}

 {\textbf{2. $k$-separability}
A state of $n$-partite subsystems is defined to be semiseparable \cite{10.5555/2011326} if it has $2-TSS(A_{i_1}A_{i_2}...A_{i_{(n-1)}};A_{i_n})$ for any choices of $i_1,...,i_n$. This means that the state $\varrho_{A_1 \ldots A_n}$ is $2-TSS$ under all 1- $(n-1)$ partitions.
Furthermore, the concept of semi-separability has been refined to include all 2-TSS under arbitrary bipartite $i- (n-i), i=1,...,n-1$ divisions, rather than considering only the $1- (n-1)$ divisions \cite{Yu_2004}. Following this line of reasoning, $k$-separable ($k\leq n$) states are defined as those satisfying all $k$-TSS criteria after the $n$-partite state being arbitrarily partitioned into $k$ parts \cite{Horodecki_2009,Hong_2021,Ananth_2015}. Clearly, when $m=k$, the conditions of $m$-TSS, is equivalent to $k$-separable.}

 {\textbf{3. m-party partial semi-separability (m-PSS) \cite{brassard2001multi}.} 
Given an $n$-partite state, one can perform a partial trace over $n-m$ parties, with $m$ parties left. The partial semi-separability ($m$-PSS) condition checks the separability of the remaining $m$ subsystems. We denote this condition as $m-PSS(A_{1};A_{2};...;A_{m})$ if the remaining subsystems are $m$-partite total semi-separable. For example, given a tripartite state $\rho_{ABC}$, $2-PSS(A:B)$ indicates that $\rho_{AB}=\sum_ip_i\rho_{A}^i\otimes\rho_{B}^i$.}

 {To summarize briefly, n-TSS indicates that quantum entanglement vanishes completely among these $n$ subsystems. $k$-separability means that entanglement which includes the number of partite less than $k$ vanishes completely, and m-PSS is equivalent to m-TSS, except that it is with respect to a subsystem after a partial trace procedure.}

The different types of separability discussed are not equivalent, both in their definitions and practical implications. Among the three classes of separability, it is evident that full $m$-partite separability is stronger than $k$-separability with $k<m$, which, in turn, is stronger than $k-TSS$ and $k-PSS$. In the context of gravitational tripartite entanglement, we will demonstrate that there are precisely three kinds of geometrical separating conditions corresponding to the three separating conditions mentioned above. Moreover, stronger conditions for separability correspond to stricter geometrical conditions, as we will demonstrate below.

{\subsection{Geometrical conditions for {$2$-TSS and $2$-PSS} structures}}
\noindent {Our ultimate goal is to identify all the geometrical structures of gravitational subregions whose states correspond to each of the separabilities described above.} Let us begin with the simplest PSS case. Consider three Rindler-convex subregions, denoted as $A, B, C$, within a gravitational system. The procedure of tracing out one of the subsystems, say $C$, is to disregard all systems except $A$ and $B$. We can then analyze the separability of $A$ and $B$ using the separation theorem. As depicted in figure \ref{tri12}(a), each pair of Rindler-convex regions is geometrically separable (in null vacuum). This corresponds to the tripartite state $\rho_{ABC}$ satisfying $PSS(A,B)$, $PSS(B,C)$, and $PSS(C,A)$.

\begin{figure}[H]
    \centering
    \includegraphics[width=0.6\textwidth]{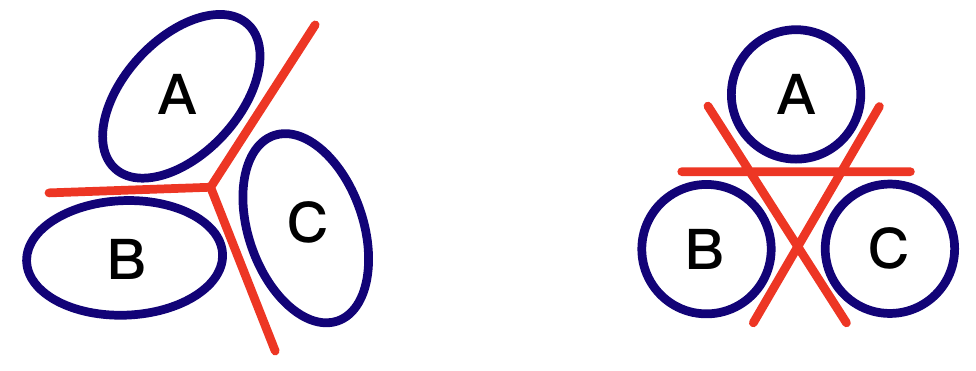} 
    \caption{In the case of three Rindler-convex regions, we can identify two typical classes of geometrically separated regions. On the left side, every pair of regions can be separated by a Rindler-convex region, but region $A$ and the combined region $BC$ cannot be separated. This corresponds to a state that satisfies the condition of being totally 2-PSS. On the right side, each region can be separated from the remaining regions by a Rindler-convex region. This type of state corresponds to a semiseparable state of the multipartite subsystems.}\label{tri12} 
\end{figure}

Generally speaking, a state of tripartite subsystems that satisfies the condition of being totally PSS could still exhibit tripartite entanglement. One famous example of such a state is the Greenberger-Horne-Zeilinger (GHZ) state
\cite{Bouwmeester_1999,kafatos2013bell,bengtsson2016brief}
\begin{equation}
    \left|G H Z\right\rangle=\frac{1}{\sqrt{2}}(|000\rangle+|111\rangle).
\end{equation}
One can check that the remaining state is separable if we trace out any subsystem from the GHZ state. Moreover, GHZ class is dense in the Hilbert space of three qubits \cite{Kim_2012}, therefore, if we pick a state randomly in Haar measure, (as we did in Page theorem argument in \cite{Ju:2023bjl}), the state must belong to the GHZ class. As a GHZ state is maximally entangled with respect to every bipartition ($A|BC$, $B|AC$, $C|AB$), it implies that the tripartite entanglement in gravitational d.o.f, if it exists, also tends to reach the maximum, like the result of Page theorem \cite{Page:1993df} in bipartite gravitational entanglement. 

Now let us analyze the $2$-TSS case in the gravitational system. First, we need to consider the geometric location of the d.o.f of $\rho_{AB}$. Although we can construct Rindler observers outside regions $A$ and $B$ individually, it is not possible to construct Rindler observers for the combined region $A\cup B$. Therefore, it is not appropriate to exclusively attribute $\rho_{AB}$ to the d.o.f within region $A\cup B$ for the sake of rigor. However, similar to the approach taken in the previous section, we can overcome this issue by making two reasonable assumptions which could be proved valid in holographic theory:

\textit{1. $\rho_{AB}$ could be reduced from $\rho_M$ with $M\supset (A\cup B)$ being Rindler convex.} This assumption implies that, as we construct the Rindler observers of $M$, we expect their ignorance to encompass all the d.o.fs of $\rho_{AB}$.

\textit{2. Given any single Rindler-convex region $N$ with $N\cap (A\cup B)=\emptyset$, we anticipate that none of the degrees of freedom within $A\cup B$ are inside $N$.} This assumption suggests that the degrees of freedom of a Rindler-convex region could be partitioned from its complement by Rindler observers.

 {It is worth noting that in a large class of spacetimes where a Rindler-convex region is the spatial part of a generalized entanglement wedge (Theorem 2.3), the above assumptions are naturally satisfied: any generalized entanglement wedge contains all d.o.f.s of smaller entanglement wedges, and two non-overlapping entanglement wedges correspond to non-overlapping d.o.f.s. This fact serves as an important example validating those assumptions.}

With these preparations, we can utilize the separation theorem to analyze the entanglement between subsystems $C$ and $AB$. If we can construct a Rindler-convex region $E$ that separates $C$ from $AB$'s Rindler convex hull $M$, there exists a region $P\supset C$ and $Q\supset M$ which renders the regions $PE$, $QE$, and $PQE$ Rindler-convex, and naturally satisfies \ref{SSA}. Using the same method as the separation theorem, one can conclude that $\rho_{ABC}$ is $TSS(C,AB)$. figure \ref{tri12} (b) depicts a state which is $TSS(A,BC)$, $TSS(B,AC)$, and $TSS(C,AB)$, i.e., semiseparable or k-separable with $k=2$.

One can discern that the geometrical condition of semiseparability in figure \ref{tri12} (b) is stronger than the geometrical condition of total 2-PSS in figure \ref{tri12} (a). This corresponds to the monogamy of bipartite squashed entanglement
\begin{equation}
    E_{\mathrm{sq}}(A;BC)\geq E_{\mathrm{sq}}(A;B)+E_{\mathrm{sq}}(A;C),
\end{equation}
where $E_{\mathrm{sq}}(A;BC)$ remains positive even though $E_{\mathrm{sq}}(A;B)$ and $E_{\mathrm{sq}}(A;C)$ becomes zero.

Is it sufficient to analyze tripartite entanglement solely through the concept of bipartite? Or equivalently, is semiseparable synonymous with fully separable? The answer is affirmative only if $\rho_{ABC}$ is a pure state; otherwise, one can construct counterexamples that are semiseparable but still having multipartite entanglement. An explicit example was presented in \cite{Bennett_1999}. Furthermore, any combination of 2-TSS and PSS cannot result in a fully separable (3-TSS) state for the mixed state $\rho_{ABC}$ \cite{brassard2001multi}. In a gravitational system, we must treat $\rho_{ABC}$ as a mixed state because we only regard the entire Cauchy slice as a pure state. Therefore, we must propose another geometrical structure that is dual to the fully separable state $\rho_{ABC}$.
\vspace{0.4cm}
\subsection{Multipartite squashed entanglement and {conditional entanglement of multipartite information}}
\noindent 
{To establish a connection between geometrical structures and multipartite entanglement structures and to generalize the separation theorem to the multipartite case, we have opted to use multipartite squashed entanglement and the conditional entanglement of multipartite information (CEMI) as the measure of multipartite entanglement.}

Multipartite squashed entanglement, which first appeared in \cite{Yang_2009,Avis_2008}, has at least two versions, $E_{\mathrm{sq}}$ and $\widetilde{E}_{\mathrm{sq}}$, as defined in \cite{Wilde_2016}
\begin{equation}\label{multiEsq}
    \begin{aligned}
& E_{\mathrm{sq}}\left(A_1 ; \cdots ; A_m\right)\equiv \frac{1}{2} \inf\left\{I\left(A_1 ; \cdots ; A_m \mid E\right)_\rho: \operatorname{Tr}_E\left\{\rho_{A_1 \cdots A_m E}\right\}=\rho_{A_1 \cdots A_m}\right\}, \\
& \widetilde{E}_{\mathrm{sq}}\left(A_1 ; \cdots ; A_m\right)\equiv \frac{1}{2} \inf \left\{\widetilde{I}\left(A_1 ; \cdots ; A_m \mid E\right)_\rho: \operatorname{Tr}_E\left\{\rho_{A_1 \cdots A_m E}\right\}=\rho_{A_1 \cdots A_m}\right\},
\end{aligned}
\end{equation}
where $I\left(A_1 ; \cdots ; A_m \mid E\right)_\rho$ and $\widetilde{I}\left(A_1 ; \cdots ; A_m \mid E\right)_\rho$ is defined as
\begin{equation}
    \begin{aligned}
    I(A_1;\cdots;A_m \mid E)_\rho &=\sum_{i=1}^m S(A_i E) - (m-1)S(E) -  S(A_1 \cdots A_m E),\\
     \widetilde{I}(A_1 ; \cdots ; A_m \mid   
     E)_\rho&=\sum_{i=1}^m S(A_{[m] \backslash\{i\}} E) - S(E) - (m-1) S(A_1 \cdots A_m E).
    \end{aligned}
\end{equation}
One can verify that each of the definitions reduces to the bipartite squashed entanglement when $m=2$. {It has been established in \cite{Davis:2018ydj} that these two measures are, in fact, equal. An intriguing observation, as demonstrated in Appendix B, is that the geometric conditions for both of these measures to vanish are also proven to be equal. This coincidence strongly suggests a bidirectional correspondence between geometric structure and squashed entanglement, rather than a one-way derivation.} 
{Given the equivalence between these two measures, we will now refer to the squashed entanglement as the first one ($E_{\mathrm{sq}}$) unless stated otherwise.}

{The measure known as the conditional entanglement of multipartite information (CEMI) is defined as follows \cite{Yang_2008}}
\begin{equation}
E_I\left(A_1: \cdots: A_m\right)_\rho \equiv \frac{1}{2} \inf _{\rho_{A_1 A_1^{\prime} \cdots A_m A_m^{\prime}}} I\left(A_1 A_1^{\prime}: \cdots: A_m A_m^{\prime}\right)_\rho-I\left(A_1^{\prime}: \cdots: A_m^{\prime}\right)_\rho,
\end{equation}
{where similar to the multipartite squashed entanglement, the infimum is over all extensions $\rho_{A_1 A_1^{\prime} \cdots A_m A_m^{\prime}}$ of $\rho_{A_1 \cdots A_m}$, i.e.,}
\begin{equation}
    \rho_{A_1 \cdots A_m}=\operatorname{Tr}_{A_1^{\prime} \cdots A_m^{\prime}}\left\{\rho_{A_1 A_1^{\prime} \cdots A_m A_m^{\prime}}\right\}.
\end{equation}
{It is evident that the non-negativity of CEMI is guaranteed by strong subadditivity. Furthermore, It has been demonstrated in \cite{Wilde_2015} by Wilde, utilizing strong subadditivity, that CEMI serves as an upper bound for the multipartite squashed entanglement.}

Are the multipartite squashed entanglement {and CEMI suitable measures?} In \cite{Guo_2020}, Guo and Zhang proposed the following properties that a multipartite entanglement measure should satisfy: 1. it vanishes in a fully separable state; 2. it cannot increase under m-partite local operations and classical communication (LOCC); 3. it reduces to the bipartite case when $m=2$; 4. it obeys monogamy \cite{Koashi_2004}. These properties have already been proven for multipartite squashed entanglement \cite{Yang_2009,Avis_2008}, which has also the qualities of convexity, subadditivity, continuity, etc. {For CEMI, the properties mentioned above also hold, except for the fact that it is not known whether it is monogamous. Our primary concern is the faithfulness of these measures in determining whether a state is separable or not. Regarding multipartite squashed entanglement, this remains an open question to our knowledge \cite{Li_2014, Li_2018}\footnote{We would like to thank M.M. Wilde for pointing this out.}. However, as the vanishing condition of CEMI is stronger than the vanishing condition of $E_{sq}$, the faithfulness of CEMI has already been proven in \cite{Wilde_2015}, i.e., $E_I\left(A_1: \cdots: A_m\right)_\rho=0$ if and only if $\rho_{A_1 \cdots A_m}$ is a full m-partite separable state (\ref{FSP}). In summary, while $E_{sq}$ and CEMI are both good measures, they have not been proven to be perfect multipartite entanglement measures yet.}

\begin{figure}[H]
    \centering
    \includegraphics[width=0.95\textwidth]{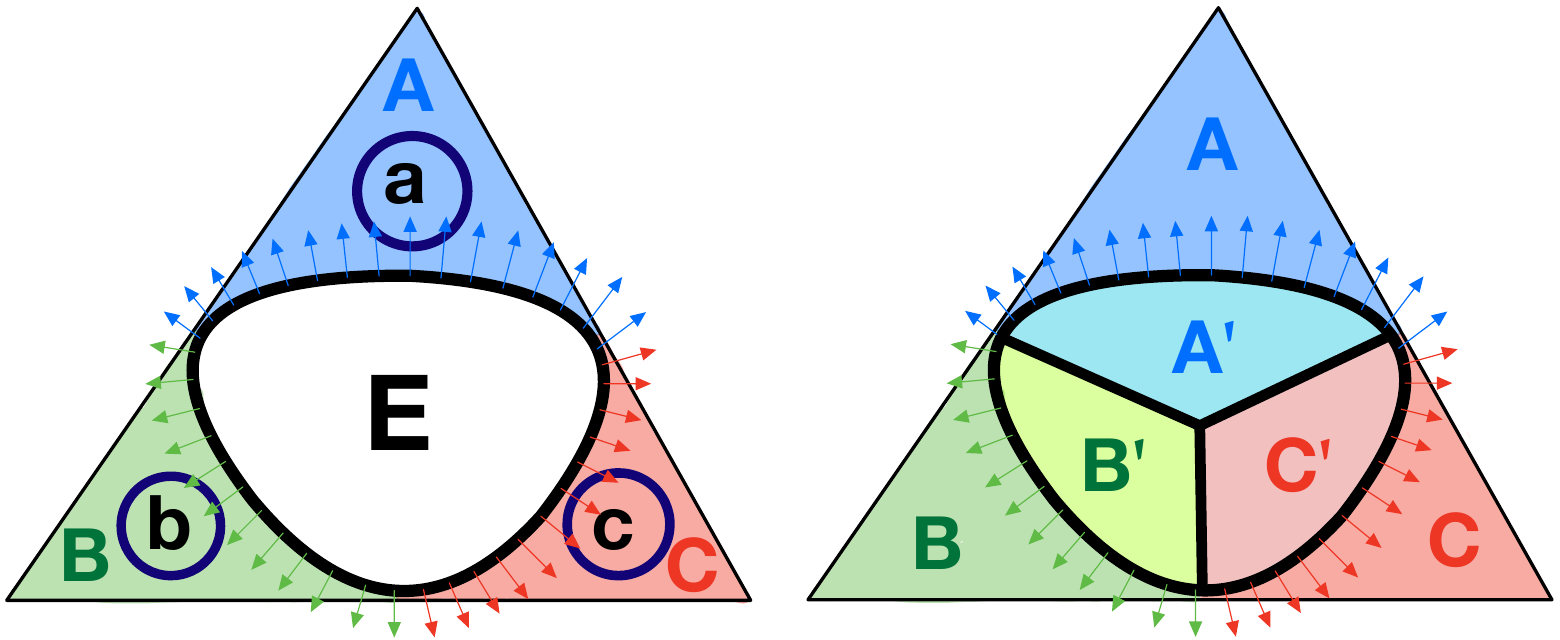} 
    \caption{{Geometric structures corresponding to vanishing multipartite squashed entanglement and vanishing CEMI. Left:} regions $E$, $AE$, $BE$, $CE$, and $ABCE$ are all Rindler-convex, and region $E$ separates Rindler-convex regions $a$, $b$, and $c$. We can divide $E$'s Rindler observers into three groups in the blue-shaded region A, the green-shaded region B, and the red-shaded region C, respectively. These three groups of observers are entirely causally disconnected from each other and can observe all the degrees of freedom within regions $a$, $b$, and $c$, respectively. {Right: in contrast to the left side, region $E$ is partitioned into three regions $A'$, $B'$, and $C'$. To ensure Rindler convexity for all of them, their interfaces must exhibit Rindler convexity on both sides, requiring that the background spacetime is a null vacuum. Since regions $A$, $B$, and $C$ are only adjacent to regions $A'$, $B'$, and $C'$, one can demonstrate, utilizing the Rindler convexity of regions $E$ and $ABCE$, that regions $AA'$, $BB'$, and $CC'$ are also Rindler convex.}}\label{TRI} 
\end{figure}

Having analyzed the quantum information aspect of the separable condition, let us now focus on its geometrical correspondence. Physically, we could learn from the observer interpretation proposed in the last section, and formulate a generalized hypothesis as follows: if there exist three groups of consistent observers, each capable of observing all the degrees of freedom in regions $A$, $B$, and $C$ respectively and causally disconnected with other groups, the state $\rho_{ABC}$ is fully 3-partite separable. Figure \ref{TRI} illustrates this case, where each group of observers is causally disconnected with other groups as they are behind others' horizon (part of the large triangle).

{In the dual boundary theory, as we reviewed in Sections 2 and 4, bulk Rindler observers that are causally disconnected from each other correspond to causally disconnected boundary observers. The boundary time cutoff then eliminates all quantum entanglement between these boundary subregions, implying the vanishing of multipartite entanglement for the corresponding bulk subregions\footnote{ {One can straightforwardly generalize Figure \ref{bcutoff} to the multipartite case in the same vein: given \(n\) regions \(A, B, C, \dots\) inside a time band such that no pair of these regions can be causally connected, one may take the union of the \(n-1\) gap regions between them as region \(E\). The \(n\)-partite conditional mutual information}  {then decomposes as the sum of CMIs—for example, when \(n=3\),
\[
I(A\!:\!B\!:\!C\mid E) = I(A\!:\!B\mid E) + I(AB\!:\!C\mid E).
\]
Due to the causal disconnection of \(A, B, C\), each CMI vanishes in this time band, resulting in the vanishing of the 3-CMI. In summary, the time cutoff also “cuts” long-range multipartite entanglement (as measured by \(n\)-CMI) in the same way it cuts bipartite entanglement.}} via the GRW subregion–subregion duality. Consequently, one expects a multipartite version of the separation theorem on the gravity side, mirroring the boundary physics in which the time cutoff destroys the entanglement structure.
}

As depicted in Figure \ref{TRI}, it is evident that there exists a Rindler convex region $E$ which effectively separates regions $a$, $b$, and $c$ from each other. This partition of the Rindler observer into three distinct sets enables each set to observe the degrees of freedom of $a$, $b$, or $c$, respectively. Subsequently, one can construct Rindler-convex regions $AE$, $BE$, $CE$, and $ABCE$. By evaluating the multipartite conditional mutual information between $A$, $B$ and $C$, and considering subsystem $E$ as the condition, we find that the tripartite squashed entanglement\footnote{{It can be proved that the multipartite squashed entanglement is non-negative via the non-negativity of the conditional mutual information: $I(A;B;C | E)_\rho=I(A;B|E)_\rho+I(AB;C | E)_\rho\geq0$.}} vanishes
\begin{equation}
    \begin{aligned}
    E_{\mathrm{sq}}\left(A;B;C\right)&=\inf \{I(A;B;C \mid E)_\rho\}\\&=S_{AE}+S_{BE}+S_{CE}-2S_{E}-S_{ABC}\\&=Area(AE)+Area(BE)+Area(CE)-2Area(E)-Area(ABC)\\&=0.
    \end{aligned}
\end{equation}
Here, we interpret $\rho_A$ as the density matrix reduced from $\rho_{AE}$ after taking the partial trace over subsystem $E$, as we emphasized in the bipartite case. Following the two assumptions above, it is straightforward to conclude that $\rho_A$ must contain all the degrees of freedom in $a$. Since $A$, $B$, $C$ have zero multipartite squashed entanglement {under the condition of $\rho_E$, it implies that $a$, $b$, $c$ also have zero multipartite squashed entanglement due to the monogamy of multipartite squashed entanglement.}

The distinction between the geometrical condition for vanishing multipartite squashed entanglement and semiseparability lies in the fact that multipartite squashed entanglement requires a single subsystem $E$ as the condition to separate all subsystems, while semiseparability demands $E_1,E_2,...,E_m$ to separate $A_i$ from $A_{[m] \backslash{i}}$ respectively. It is clear that the former is a stricter condition. Correspondingly, in the geometric case, semiseparability requires different Rindler-convex regions to separate region $A_i$ from others, while full m-partite separability demands a single Rindler-convex region to separate all $A_m$ regions. As depicted in figure \ref{TRI}, this geometrical condition is noticeably stricter than the semiseparable case.

{As the faithfulness of multipartite squashed entanglement has not been proven yet, we cannot ensure that the $m$-partite state in the gravitational system is a full $m$-partite separable state. That motivates us to analyze the CEMI between gravitational subregions. As shown in figure \ref{TRI}, one can partition convex region $E$ into three parts $A'$, $B'$, and $C'$, with each part only adjacent to region $A$, $B$ and $C$, respectively. If we demand that regions $A'$, $B'$ and $C'$ are all Rindler-convex, the background spacetime must be a null vacuum to allow the interfaces to be Rindler convex on both sides.}

{In a null vacuum spacetime, the multipartite mutual information of $m$ Rindler convex regions $A_1$,$A_2,..., A_m$ is twice the area of their interfaces if region $A_1A_2...A_m$ is also Rindler convex, similar to its bipartite counterpart \cite{Ju:2023bjl}}
\begin{equation}
\begin{aligned}
I(A_1:A_2:...:A_m)&=S_{A_1}+S_{A_2}+...+S_{A_m}-S_{A_1A_2...A_m}\\
    &=\sum_{i=1}^m Area(\partial A_i)-Area(\partial (A_1A_2...A_m))\\
    &=2\sum_{i>j} Area((\partial A_i) \cap (\partial A_j)).
\end{aligned}
\end{equation}
{With this preparation, we can calculate the CEMI of gravitational subregions $A$, $B$, and $C$ in figure \ref{TRI} as follows}
\begin{equation}\label{CEMI}
\begin{aligned}
    E_I(A:B:C)_\rho &\equiv \frac{1}{2} \inf _{\rho_{A A^{\prime}B B^{\prime}C C^{\prime}}} I(A A^{\prime}:B B^{\prime}:C C^{\prime})_\rho-I(A^{\prime}:B^{\prime}:C^{\prime})_\rho,\\
    &=2(A_{AA'|BB'}+A_{AA'|CC'}+ A_{BB'|CC'}-A_{A'|B'}-A_{A'|C'}- A_{B'|C'})\\
    &=0,
\end{aligned}
\end{equation}
{where $A_{M|N}$ represents the area of the interface between regions $M$ and $N$. CEMI vanishes due to the fact that adding subsystems $A$, $B$, and $C$ in multipartite mutual information $I(A^{\prime}:B^{\prime}:C^{\prime})$ does not modify the area of interfaces between them.}

{In the second line of equation (\ref{CEMI}), we require regions $A'$, $B'$, $C'$, $A'B'C'$, $AA'$, $BB'$, $CC'$, and $AA'BB'CC'$ to be Rindler-convex. Since region $A'B'C'$ is region $E$, in comparison to the multipartite squashed entanglement, we only need to additionally ensure that the background spacetime is a null vacuum so that $A'$, $B'$, and $C'$ can all be Rindler-convex regions. This corresponds to the fact in quantum information theory that CEMI serves as an upper bound of $E_{sq}$, making its vanishing condition stronger than that of $E_{sq}$. Due to the faithfulness of CEMI, we can conclude that the state of subsystems $A$, $B$, and $C$ is fully $3$-partite separable. This proof can be naturally generalized to $m$-partite cases.}

\textbf{Gravitational Multipartite Separation Theorem.} The quantum state $\rho_{A_1,...,A_m}$ of gravitational subregions $A_1,A_2,...,A_m$ {has vanishing multipartite squashed entanglement} if there exists a Rindler-convex region $E$ that simultaneously separates regions $A_1,...,A_m$ from each other, {and it is fully $m$-separable in the null vacuum.}

On the holographic field theory side, employing similar arguments as discussed in the previous section, one can derive the corresponding sufficient geometric condition for constructing a $\tilde\rho$ state with boundary subregions $A_1, A_2, ..., A_m$ as full m-partite separable while preserving $\rho_{A_1}...\rho_{A_m}$. 
Specifically, this condition entails the existence of $m-1$ spherical shells that separate regions $A_1, ..., A_m$ from each other, resulting in the same explanations for boundary time cutoff, etc.

By now, {we have established the multipartite separation theorem on the gravity side, corresponding to the physics of time cutoff {eliminating IR entanglement structure} on the boundary side in vacuum AdS. In a null vacuum, our hypothesis put forth by observer physics has been proven. Compared with the bipartite case, the only difference is that we require region $E$ to separate three regions simultaneously.  {On the other hand, in non-vacuum cases the separation theorem must be reexamined on both sides. In the bulk, CEMI may no longer be computable because \(A'\), \(B'\), and \(C'\) might not all be Rindler-convex simultaneously. Furthermore, the faithfulness of multipartite squashed entanglement remains an open question, casting doubt on the separability of gravitational multipartite states. On the boundary side, GRW subregion duality can fail, and causally disconnected regions may nonetheless exhibit nonzero conditional mutual information when evaluated via the RT formula. These issues render the physical interpretation of a boundary time cutoff unclear in non-vacuum settings.
}}

With the geometric condition corresponding to the fully m-separable state {in a null vacuum} established, all geometric conditions respectively corresponding to $k-TSS$, $k-PSS$ and $k-separable$ for any $k$ can be easily identified because their definitions are based on full m-separable, as we have already demonstrated. Given the fact that `simultaneously separating' is a stricter geometric condition than `respectively separating', one can easily show that the geometric condition corresponding to the stronger separability condition of quantum states is also stronger. This gives us a clearer implication on the correspondence of separation itself, i.e., the slogan we summarized in the last section (\ref{slogan}).

\subsection{Probing quantum entanglement}

\noindent \textbf{Locality of gravitational entanglement.} 
As we previously discussed, in the case of null vacuum, bipartite quantum entanglement between subregions is local and vanishes immediately when two non-adjacent Rindler-convex regions are considered. However, in the context of multipartite entanglement, the situation is different. The global tripartite entanglement may not vanish until we physically separate the three regions from each other by a finite distance to allow for the existence for Rindler-convex region $E$. This observation indicates that multipartite gravitational entanglement can carry quantum information over long distances. Moreover, the distance over which the entanglement persists increases with the number of parties involved. In other words, the more parties there are, the longer the separation distance is needed to fully disentangle them.  This finding aligns with the difficulty of simultaneously separating a large number of subregions using a single Rindler-convex region, as it becomes increasingly challenging as the number of parties grows.

When a time cutoff is introduced, multipartite entanglement is affected in a similar way to the case of bipartite entanglement. In the extreme case where no observers can be causally connected with any of the multiple regions, making it impossible to detect the entanglement structures between them, the combination of the entire set of regions becomes Rindler-convex. In this scenario, the entropy of the combined system is proportional to its surface area. As a result, the multipartite mutual information $I$ vanishes
\begin{equation}
    I(A_1:A_2:...:A_n)=S_1+S_2+...+S_n-S_{12...n}=0.
\end{equation}
The entire state is a simple product state without any quantum or classical multipartite correlations \cite{Avis_2008}
\begin{equation}\label{SSP}
    \rho_{A_1 \ldots A_m}=\rho_1 \otimes \ldots \otimes \rho_m.
\end{equation}
What we are primarily interested in is the quantum entanglement structure rather than classical correlations. Here, we can investigate the quantum entanglement structure by analyzing when the time cutoff removes all quantum entanglement. As observers can detect the entanglement structures within the ``time layer", which refers to the spacetime region between the cutoff, the information regarding the quantum entanglement structure is indeed embedded in the geometric shape of the minimal time cutoff required to remove all quantum entanglement.

Here, we present a simple example to investigate the tripartite entanglement critical distance between three spherical regions in a 2+1-dimensional null vacuum geometry. To facilitate our analysis, we introduce a ``planar" time cutoff in the spacetime, which leaves a ``time layer" between the regions. In figure \ref{trisep}, it is evident that separating region $abc$ is easier than where there is no cutoff, as demonstrated. Specifically, if $a=1.1$, the thickness of the time layer must satisfy $\Delta t \lesssim 10.87$ to render $abc$ semiseparable, and $\Delta t \lesssim 3.88$ to eliminate the tripartite entanglement completely. 
\begin{figure}[H] 
    \centering
    \includegraphics[width=0.35\textwidth]{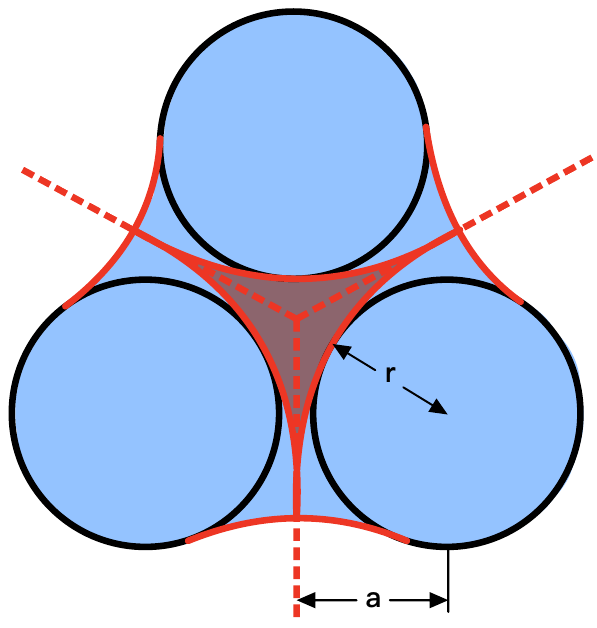} 
    \caption{We observe three black disk regions $abc$ with a radius of $r=1$, positioned at a distance of $2a$ from each other. The introduction of a planar time cutoff is depicted by the red curves, representing the largest lightspheres associated with each region. The Rindler-convex hull of $abc$ is highlighted in blue shading, while the brown shading represents region $E$, which acts as a separator for the Rindler-convex hull. As a result, the tripartite quantum entanglement of $abc$ vanishes.}\label{trisep}
\end{figure}

{In addition, if a time cutoff is introduced in a non-vacuum spacetime geometry, it can render some surfaces Rindler-convex on both sides due to the relaxation of Rindler-convexity. This makes CEMI calculable, and theoretically, though complex, one can analyze the full m-partite separability in this scenario. This could also be viewed as an approach to eliminate {gravitational} multipartite quantum entanglement by {bulk} time cutoff.}

Based on the discussions presented earlier, we can draw the following conclusions when a time cutoff is introduced in a {null vacuum} gravitational system.

1.If \textbf{no} observer can be causally connected with more than one of the subregions, the state on the subregions must be a simple product state (\ref{SSP}).

2.If there \textbf{are} n sets of observers, each being causally connected with only one of the subregions, the state on the subregions must be a full m-partite separable state (\ref{FSP}).

These observations highlight the difference of condition between eliminating mutual information and eliminating quantum entanglement {(multipartite squashed entanglement or CEMI)} via an observer approach.

\section{Conclusion and discussion}
\noindent
In this paper, we examine the {bipartite and multipartite} entanglement structures of gravitational {subregions and their dual boundary quantum states}. Through an analysis of squashed entanglement {and CEMI} between disjoint subregions, we establish that quantum states on geometrically separated subregions satisfying a certain separable condition must exhibit vanishing multipartite squashed entanglement {or CEMI} at the semiclassical order. This implies the absence of quantum entanglement between these subregions, signifying that the entire state is fully $m$-partite separable {in a null vacuum geometry}. Based on this, we {associate} various separabilities in quantum information theory, such as $m$-party total semiseparable and $k$-separable, to well-defined geometrical structures in gravitational theory, respectively.

On the dual boundary field theory side, through the utilization of subregion-subregion duality and the GRW subregion duality, we have determined that the multipartite separation theorem assesses the separability of distinct boundary subregions in the $\tilde\rho$ state (which corresponds to a GRW in the bulk). We have identified the sufficient condition for $\tilde\rho$ to exhibit full $m$-partite separability among boundary subregions $A_1, \ldots, A_m$, and this condition naturally arises in boundary physics with the introduction of a time cutoff.

Finally, we introduce a time cutoff in gravitational theory to eliminate correlations between subregions. We argue that the relaxation of the condition of geometrical separabilities corresponds to the elimination of quantum entanglement under time cutoff. This insight inspires us to consider time cutoff itself as a probe of the location of quantum entanglement. We provide an example to illustrate this concept. 

Given the striking similarity between the properties of the geometric separability and those of the separabilities of quantum states, it becomes difficult to dismiss this as mere coincidence. Moreover, the geometric structure can serve as a pedagogic graphical representation to identify entanglement structures for complex quantum states. This is reminiscent of the well-known graph in \cite{bengtsson2016brief}, which uses the topological structure of rope knots to illustrate the difference between GHZ and W states. This concept motivates us to thoroughly analyze the explicit geometrical conditions of separabilities, hoping to uncover constructive insights into entanglement structures in the realm of quantum information theory. In the following, we present some intriguing observations regarding geometrical separabilities.

\textbf{Dimension-partite relationship in the multipartite entanglement in higher dimensions.}
The observation that the dimensionality of the Cauchy slice affects the ability to separate subregions in gravitational systems is an interesting finding. It suggests that higher-dimensional Cauchy slices provide more possibilities for observers to accelerate in different directions without having causal connections with other sets of observers, making the separation of subregions easier.

In low-dimensional gravitational systems, such as those with only one spatial dimension, it is not possible to separate tripartite regions with three sets of observers. However, in two spatial dimensions, we can construct geometric structures corresponding to 2-partite separable (2-PSS), 2-tripartite separable (2-TSS), and 3-tripartite separable (3-TSS) states.

When it comes to 4-partite entanglement, the situation becomes more complex. If the spatial dimension is limited to two dimensions, it is insufficient to accommodate all distinct geometric configurations corresponding to the full range of inequivalent separable types of 4-partite entanglement. Specifically, the geometric condition for semiseparable is equivalent to 3-separable and fully separable on a 2-dimensional Cauchy slice. This does not imply that the separation theorem is inconsistent, but rather suggests that low-dimensional gravitational systems may lack some geometric structures corresponding to certain multipartite entanglement configurations.

However, it can be proven that if the spatial dimension is not less than (m-1), where m is the number of parties involved, there exist distinct geometric structures corresponding to each type of separable quantum states (n-PSS and n-TSS states). This observation establishes a dimension-partite relationship, indicating that the dimensionality of the Cauchy slice plays a crucial role in determining the available geometric structures associated with multipartite entanglement.

Further investigation of this dimension-partite relationship would be an interesting avenue for future research, as it may provide deeper insights into the interplay between geometry and multipartite entanglement in gravitational systems.

\textbf{A possible global entanglement priority rule.}
In gravitational systems, we can observe certain patterns and strategies regarding multipartite entanglement. Specifically, in the tripartite case, we have encountered systems that exhibit 2-PSS for all partitions, yet they are not semiseparable. This suggests a higher likelihood of GHZ-type states appearing compared to W-type states \cite{bengtsson2016brief}.

Furthermore, even if a state is semiseparable, it may not be fully separated. To illustrate this, let us consider a group of regions in general spacetime initially close to each other. As we gradually increase the distance between them, we observe a progression: the regions first become 2-PSS with each other, resulting in the decrease and eventual disappearance of bipartite entanglement. This leads to a semiseparable state. Continuing this process, the tripartite entanglement within each division of regions vanishes, resulting in a 3-separable state. Subsequently, the state transitions into being 4-separable, 5-separable, and so on, until it becomes fully separable. Importantly, this deformation process is continuous.

Based on these observations, it appears that entanglement in gravitational systems tends to prioritize entanglements involving more parties over those involving fewer parties. In other words, global entanglement is more prevalent and dominant in these systems. If we consider the entire state as a random mixed state, we would expect states with more ``global" entanglement to have a larger measure in the mixed state Hilbert space. Specifically, states like W class, where the multipartite entanglement originates solely from lower-partite entanglements ($n\leq m-1$), may have a measure of zero.

To summarize, our observations reveal patterns in multipartite entanglement in gravitational systems. These include the prevalence of GHZ-type states, the continuous deformation of entanglement as distances increase, and the prioritization of entanglements involving more parties. These findings provide insights into the distribution and dynamics of entanglement in gravitational systems, as well as its relation to measures of mixed state Hilbert spaces.
    
\section*{Acknowledgement}
    
We thank Teng-Zhou Lai, Yi-Yu Lin and Chang-Shui Yu for useful discussions. We would like to thank {Mark M. Wilde} for helpful email correspondence and comments. This work was supported by Project 12035016 and 12275275 supported by the National Natural Science Foundation of China.
It is also supported by Beijing Natural Science Foundation under Grant No. 1222031.

\appendix
\section{Examples of Rindler-convex surfaces}
{In order to prevent the above discussion from being overly abstract, here we provide some simple and specific examples of GRWs and Rindler convex surfaces on a Cauchy slice with zero extrinsic curvature in various background geometries as follows.}
   \begin{itemize}

       \item {\it A Rindler wedge in AdS or in Minkowski spacetime is a GRW}. {In Rindler wedge, Rindler transformation guarantees that a set of consistent accelerating observers exist in the Rindler wedge with the Rindler horizon being their horizon.} By the definition of Rindler-convexity, a Rindler wedge is a GRW.

       \item {\it In Minkowski spacetime, on a flat Cauchy slice, surface $\partial A$ is Rindler-convex if and only if $A$ is a convex set. }As we have proved, in flat spacetime, Rindler convexity reduces to the usual notion of convexity. Therefore in Minkowski spacetime, any convex spatial surface could be the bifurcation surface of the horizon of well-defined accelerating observers.

       \item {\it In de-Sitter vacuum, the static patch is the smallest GRW with the cosmological horizon being its Rindler-convex surface, which has the largest surface area}, as described in Figure \ref{dssphere}.
       \item {\it A spherical spatial region \cite{Balasubramanian:2013rqa} outside the trapped/anti-trapped surface in a spacetime with spherical symmetry is Rindler-convex}, as described in Figure \ref{dssphere}.
       \item {\it The intersection of a killing horizon with a Cauchy slice with zero extrinsic curvature is Rindler-convex.} {As there are no caustics on a killing horizon, thus normal condition of Rindler convexity is satisfied.}
       
       \item {\it In AdS vacuum or a BTZ black hole geometry, the causal wedge ($CW(A)$) of a boundary subregion $A$ is a GRW, while its entanglement wedge ($EW(A)$) is generally not a GRW (Section 4).}

   \end{itemize}

\begin{figure}[h]
    \centering 
    \includegraphics[width=0.7\textwidth]{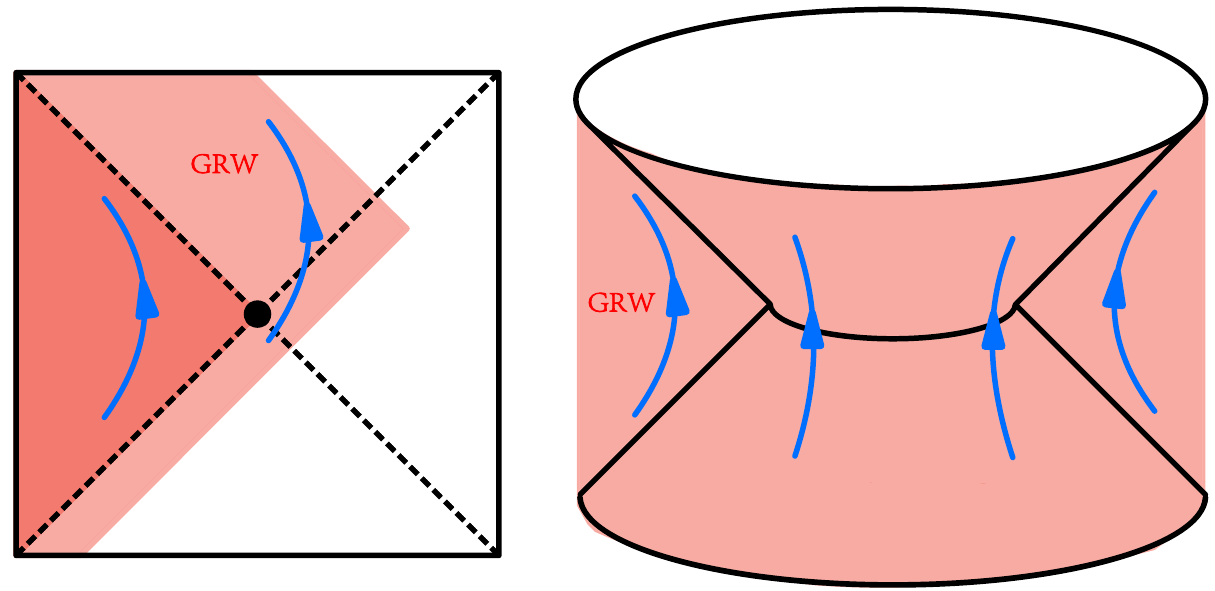} 
    \caption{Left: in dS spacetime, the static patch is the smallest GRW with the largest surface area, which is the cosmological horizon. A larger GRW is also plotted as the light red shaded region in the figure, which has a smaller surface area, though.  Right: in a spherically symmetric spacetime, the radially radiating null geodesics will never intersect if it is not inside an trapped/anti-trapped surface. Thus, a spacial sphere is always Rindler-convex under this circumstance.}
    \label{dssphere}
\end{figure}

Besides these simple examples, we could also produce Rindler convex regions through the following procedure.
\begin{theorem}
    \textit{When $T_{\mu\nu}k^\mu k^\nu=C_{\rho\mu\nu\sigma}k^\rho k^\sigma=0$, on a Cauchy slice with zero extrinsic curvature $K_{\mu\nu}=0$, Rindler-convexity is equivalent to geodesic convexity, where $T_{\mu\nu}$ is the energy-momentum tensor, $C_{\rho\mu\nu\sigma}$ is the Weyl tensor, and $k^\mu$ is any null vector. }
\end{theorem}
    \begin{proof}
        The Proof can be found in appendix A in \cite{Ju:2023bjl}. 
    \end{proof}
\begin{theorem} \label{CAPtheorem}
    \textit{If $A$ and $B$ are Rindler convex regions on a Cauchy slice, and $A\cap B$ has a smooth boundary, then $A\cap B$ is also a Rindler convex region.}
\end{theorem}
 \begin{proof}
     {Given $C=A\cap B$, we have:}
    \begin{equation}
        \partial C=(\partial C \cap \partial A )\cup(\partial C \cap \partial B )
    \end{equation}
    because of the Rindler-convexity of regions $A$ and $B$, and due to the tangential condition, the lightspheres tangential to $\partial A$ ($\partial B$) will never reach the inside of $A$ ($B$). As a result, the lightsphere tangential to $\partial C$ that is either tangential to $A$ or $B$ will never reach the inside of $A\cap B$, which makes $C$ a Rindler-convex region by the tangential condition.
 \end{proof}

    Note that any smooth Rindler-concave surface could also look like a horizon locally by a coordinate transformation, however, as we have shown explicitly, there cannot exist globally well-defined accelerating observers for concave subregions as the worldlines of the accelerating observers would intersect at a spacetime point indicating that the observer at that point cannot be uniquely defined, thus the degrees of freedom in a concave subregion cannot be separated from its complement in a consistent way. This Rindler-convexity condition is a {\it global} condition and this may also imply that the entropy is also related to the global structure of the horizon. 
    
    \subsection{Explicit geometric constructions of Rindler-convex regions in various spacetime}
    \noindent Given a specific Cauchy slice in a specific spacetime, one would like to {to obtain explicitly geometry constructions of Rindler-convex surfaces and corresponding generalized Rindler wedges \cite{Ju:2023bjl}.} In principle, we can find all the lightspheres on this Cauchy slice and then use the tangential condition to determine the Rindler-convexity condition. However, as a larger lightsphere always can contain a smaller one due to causality, the tangential condition for a larger lightsphere is always stronger than the tangential condition for a smaller lightsphere. As a result, we only have to use the infinitely large lightsphere emitting from (converging to) past (future) null infinity to test the Rindler convexity.

In the following, we construct explicit Rindler convex surfaces in three types of geometries, the $2+1$ dimensional conformally flat spacetime with the flat spacetime as a special case, the $2+1$ dimensional conformally flat spacetime with conformal boundaries with pure AdS being a special case, and the four dimensional Schwarzschild black hole spacetime.

\noindent{\bf Case I}: a conformally flat spacetime without a spacetime boundary in the conformal coordinate.
The spacetime metric is:
\begin{equation}
ds^2=\Omega^2(t,x,y)(-dt^2+dx^2+dy^2),\quad\quad (t,x,y \in \mathbb{R})
\end{equation}
where $\Omega^2(t,x,y)$ is a smooth positive function without singular values. {This includes the flat spacetime as a special case.} This spacetime could be viewed as a Weyl transformation acting on a 3D Minkowski spacetime. As a Weyl transformation preserves null geodesics \cite{Carroll:2004st}, null geodesics are \textit{straight lines} in conformal coordinates. On the $t=0$ Cauchy slice, lightspheres are circles, and the infinitely large lightspheres are straight lines. Due to the tangential condition, Rindler-convexity regions on the $t=0$ Cauchy slice are equal to the convex sets in the 2D flat Euclidean geometry of $x$ and $y$. We parametrize the Rindler convex surface, i.e. the boundary of a (strictly) convex region $A$, denoted as $\partial A$, to be $(x=x_0(\theta), y=y_0(\theta))$. The Rindler convexity condition requires that they satisfy the parametric equations as follows
\begin{equation}\label{parametric}
        \left\{
        \begin{aligned}
        x & = x_0(\theta)  \\
        y & = y_0(\theta) 
        \end{aligned}
        \right.,\quad\text{where} \quad \frac {dx_0(\theta)}{d\theta}\cos\theta +\frac {dy_0(\theta)}{d\theta}\sin\theta=0,\quad(\theta \in [0,2\pi)).
    \end{equation}
In this equation, $\theta$ denotes the azimuth of the vector normal to the convex curve. The convexity is reflected in the fact that $x_0(\theta)$ and $y_0(\theta)$ are single-valued functions.

The procedure to obtain any Rindler convex surface from this equation is the follows. We could choose an arbitrary function of $x(\theta)$ which is periodic in $theta$ and then obtain the corresponding $y(\theta)$ by solving the differential equation in (\ref{parametric}). In this way, all Rindler convex surfaces could be produced in this spacetime. The simplest example would be to choose $x(\theta)= \cos \theta$ and $y(\theta)$ could be solved to be $c+\sin\theta$, where $c$ is an arbitrary integration constant. This gives a circle in the conformally flat spacetime. We can check the shape of the surface is always convex by choosing arbitrary periodic functions of $x(\theta)$.

   Let us prove the convexity of $\partial A$ using both tangential condition and normal condition as follows. The equation of the infinitely large lightsphere tangential to $\partial A$ at $\theta_0$ is
\begin{equation}\label{support}
    y_0'(\theta_0)(x-x_0(\theta_0))-x_0'(\theta_0)(y-y_0(\theta_0))=0.
\end{equation}
Using equation (\ref{parametric}), one can prove that 
\begin{equation}
    y_0'(\theta_0)(x_0(\theta)-x_0(\theta_0))-x_0'(\theta_0)(y_0(\theta)-y_0(\theta_0))\leq0,
\end{equation}
where the equality holds if and only if $\theta=\theta_0$, which means that every point on $\partial A$ is on the same side of the infinitely large lightsphere.  In other words, line (\ref{support}) serves as the supporting hyperplane of the convex region $A$. Consequently, $\partial A$ is a Rindler-convex surface due to the tangential condition.

Testing Rindler-convexity via the normal condition is straightforward. In conformal coordinates, the null geodesics normal to $\partial A$ at $(0,x_0(\theta_0),y_0(\theta_0))$ are given by
\begin{equation}
     \left\{
        \begin{aligned}
        t & = t_0 \quad (t>0)\\
        x & = x_0(\theta) + t_0\cos\theta\\
        y & = y_0(\theta) + t_0\sin\theta 
        \end{aligned}
    \right. \quad\text{and}\quad
    \left\{
        \begin{aligned}
        t & = t_0 \quad (t<0)\\
        x & = x_0(\theta) - t_0\cos\theta\\
        y & = y_0(\theta) - t_0\sin\theta 
        \end{aligned}
    \right..
\end{equation}
These geodesics point to the future and the past, respectively. Let us choose two arbitrary null geodesics pointing to the future (without loss of generality) with $\theta=\theta_1$ and $\theta=\theta_2$ respectively. One can calculate the distance between two points $d(t_0)$ on these geodesics with the same $t_0$. This distance is a monotonically increasing function of $t_0$, which means $d(t_0)\geq d(0)>0$. In other words, the normal null geodesics never intersect, verifying the normal condition.

\noindent{\bf Case II}: conformally flat spacetimes with conformal boundaries.
Conformal boundaries can generally exist in conformally flat spacetimes. The most well-known case is AdS spacetime,  in Poincaré half-plane coordinates, the metric given by:
\begin{equation}
ds_1^2=\frac1{z^2}(-dt^2+dx^2+dz^2),\quad\quad (t,x \in \mathbb{R}; z>0)
\end{equation}
  Here, the conformal boundary is located at $z=0$. 
As we argued in case I that null geodesics are still straight lines in conformal coordinates, the light-spheres will still be circles. However, the ``infinitely large lightsphere" is affected by the presence of the conformal boundary, which in turn, affects Rindler-convexity. Specifically, the ``infinitely large lightsphere" would be generated by the light cone whose vertex is on the conformal boundary.
From the perspective of the normal condition, the existence of a spacetime boundary relaxes the Rindler-convexity condition by making the null geodesics more difficult to intersect. Note that normal condition and tangential condition are still equivalent.

As shown in Figure \ref{examples}, with a conformal boundary existing, the infinitely large lightsphere is represented by the red semicircle. Coincidentally, it forms a geodesic in AdS spacetime, implying that Rindler convexity is equivalent to geodesic convexity in global AdS spacetime (Appendix A). 
{If we make the same ansatz as in the previous case, where we replace $y$ in (\ref{parametric}) by the $z$ coordinate, we can find that the surface could be Rindler convex even if $x_0(\theta)$ and $y_0(\theta)$ are not single-valued functions in certain cases. An example is shown as the right side of Figure \ref{Poincarepatch}. In order to get the universal form of Rindler-convex surfaces, we should use another coordinate $\lambda$ to reparameterize $x_0$ and $z_0$ functions as follows.}
\begin{equation}
    \left\{
    \begin{aligned}
    x & = x_0(\lambda)  \\
    z & = z_0(\lambda) 
    \end{aligned}
    \right.,\,\,\text{where}\,\, 
    \left\{
    \begin{aligned}
    &\lambda=\theta\,\, \text{for}\,\, {\theta}\in (0,\pi)  \\
    &\tan \lambda=x_0(\theta)-z_0(\theta)\cot \theta,\,\,\text{for}\,\,\theta \in [\pi,2\pi]. 
    \end{aligned}
    \right.
\end{equation}
Then, the surface could be Rindler-convex if and only if $x_0$ and $z_0$ are single-valued functions of $\lambda$, where $\tan \lambda$ is the $x$ coordinate of the intersection point between the normal null geodesic and the conformal boundary.

Moreover, it is worth noting that the Rindler-convexity in Poincaré coordinate is not equivalent to the Rindler-convexity in global AdS coordinate. The later is equivalent to geodesic convexity but the former is not. 

\begin{figure}[H]
        \centering 
        \includegraphics[width=0.6\textwidth]{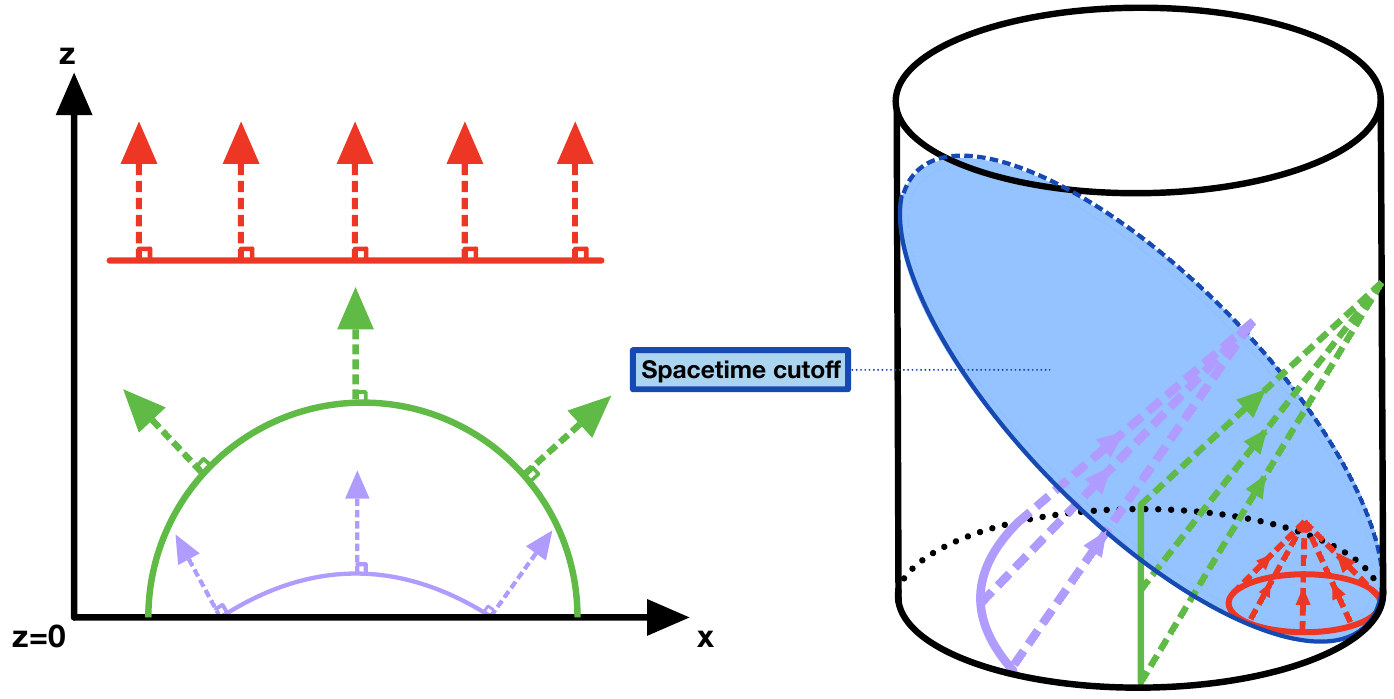} 
        \caption{Left: static Cauchy slice in Poincar\'e coordinate. Right: $t_{global}>0$ half of Poincar\'e patch from global AdS with a null cutoff (blue surface). The red and purple surfaces towards the $z>0$ side are Rindler-convex in Poincar\'e patch while Rindler-concave in global AdS.}
        \label{Poincarepatch} 
    \end{figure}

{As shown in Figure \ref{Poincarepatch}, the red surface is Rindler-convex in Poincare coordinate but Rindler-concave in global AdS.} 
This disparity arises because the Poincaré patch in global AdS spacetime has an additional spacetime boundary at $z=\infty$ and two null hypersurfaces, which relax the conditions for Rindler convexity.

\begin{figure}[H]
    \centering 
    \includegraphics[width=0.95\textwidth]{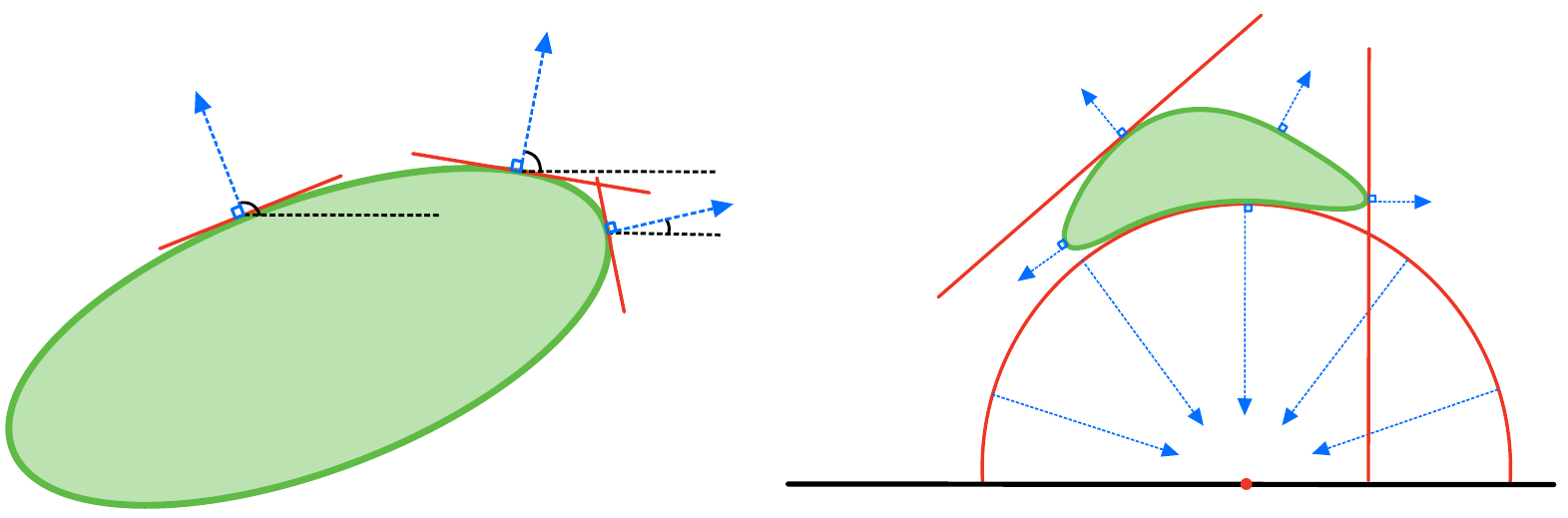} 
    \caption{Examples of Rindler-convex regions (green shaded) are shown on a flat Cauchy slice in flat spacetime (left) and in the AdS Poincaré patch (right). Blue arrows represent null geodesics, and red lines (curves) depict infinitely large lightspheres.
    } \label{examples}
\end{figure}

    \noindent{\bf Case III:}{ Rindler-convex regions in Schwarzschild black hole spacetime.}

    We construct Rindler-convex surfaces in the case of a four-dimensional Schwarzschild black hole in flat spacetime. 
    As previously mentioned, the infinitely large lightsphere is used to test the Rindler-convexity in curved spacetimes.
    To obtain the trajectories of infinitely large lightspheres, in the following we compute the null geodesics emanated from a point source far away from the Schwarzschild black hole.
    Only the light ray propagating in the equatorial plane is needed because of the rotational symmetry around the line connecting the black hole and the source.
    The metric of the equatorial plane is
    \begin{equation}
    ds^2=-f(r)dt^2+\frac{1}{f(r)}dr^2+r^2d\phi^2,  \;\; f(r)=1-2M/r
    \end{equation}
    A null geodesic $x^\mu(s)$ parameterized by affine time $s$ obeys
    \begin{equation}\label{LS0}
    g_{\mu \nu}\frac{dx^\mu}{ds}\frac{dx^\nu}{ds}=-f(r)\left(\frac{dt}{ds}\right)^2+\frac{1}{f(r)}\left(\frac{dr}{ds}\right)^2+r^2\left(\frac{d\phi}{ds}\right)^2=0.
    \end{equation}
    Two conserved quantities, energy $E$ and angular momentum $L$, are associated to the photon traveling along the geodesic with
    \begin{equation}\label{LS1}
    E=f(r)\frac{dt}{ds},  \;\;\;\; L=r^2 \frac{d\phi}{ds}.
    \end{equation}
    Thus the null geodesic equation \eqref{LS0} has the form
    \begin{equation}\label{LS2}
    -E^2+f(r)\left(\frac{dr}{ds}\right)^2+L^2 \frac{f(r)}{r^2}=0.
    \end{equation}

    Near the light source positioned at $r_0$ which is far away from the black hole ($r_0\to\infty$ and $f(r_0)\to1$), it is found that
    \begin{equation}
    -E^2+p_r^2+p_\bot^2=0,
    \end{equation}
    where $p_r=\frac{dr}{ds}$ is the radial component of momentum at infinity and $p_\bot=\frac{L}{r_0}$ can be identified as the component of momentum in the $\phi$ direction which is perpendicular to the radial direction.
    The ratio of $p_\bot$ and $p_r$ represents the direction of photon emission and is related to the impact parameter $\lambda=\frac{L}{E}$ as 
    \begin{equation}
    \frac{p_\bot}{p_r}=\frac{\lambda^2}{r_0^2-\lambda^2}.
    \end{equation}

    Therefore, $\lambda$ labels the direction of a light ray emanated from a fixed distant source. The trajectories of photons and lightspheres in a Cauchy slice of time $t$ could be computed using \eqref{LS1} and \eqref{LS2}.
    For our purpose, we use the coordinate $t$ (instead of $s$) to indicate the time a photon travels along a specific geodesic and $\lambda$ to specify which geodesic it is on.
    Combining \eqref{LS1} and \eqref{LS2}, the set of equations determining $r(t,\lambda)$ and $\phi(t,\lambda)$ is found to be 
    \begin{equation}
    \left \{ \begin{matrix}
    \frac{d\phi}{dt}=\lambda \frac{f(r)}{r^2} \\
    \left(\frac{dr}{dt}\right)^2=f^2(r)\left(1-\lambda^2\frac{f(r)}{r^2}\right)  
    \end{matrix} \right. 
    \end{equation}
    and the initial condition is $r(0,\lambda)=r_0$ for arbirary $\lambda$.
    We solve this set of equations numerically.
    The photon trajectory of impact parameter $\lambda*$ is then given by $(r(t,\lambda*),\phi(t,\lambda*)$ with $0<t<\infty$, whereas the lightsphere at time $t*$ is given by $(r(t*,\lambda),\phi(t*,\lambda))$ with $-\infty<\lambda<\infty$.
    The photon trajectories and lightspheres of a distant source are shown in Figure \ref{LS_Sch}.

    \begin{figure}[H]
        \centering 
        \includegraphics[width=0.8\textwidth]{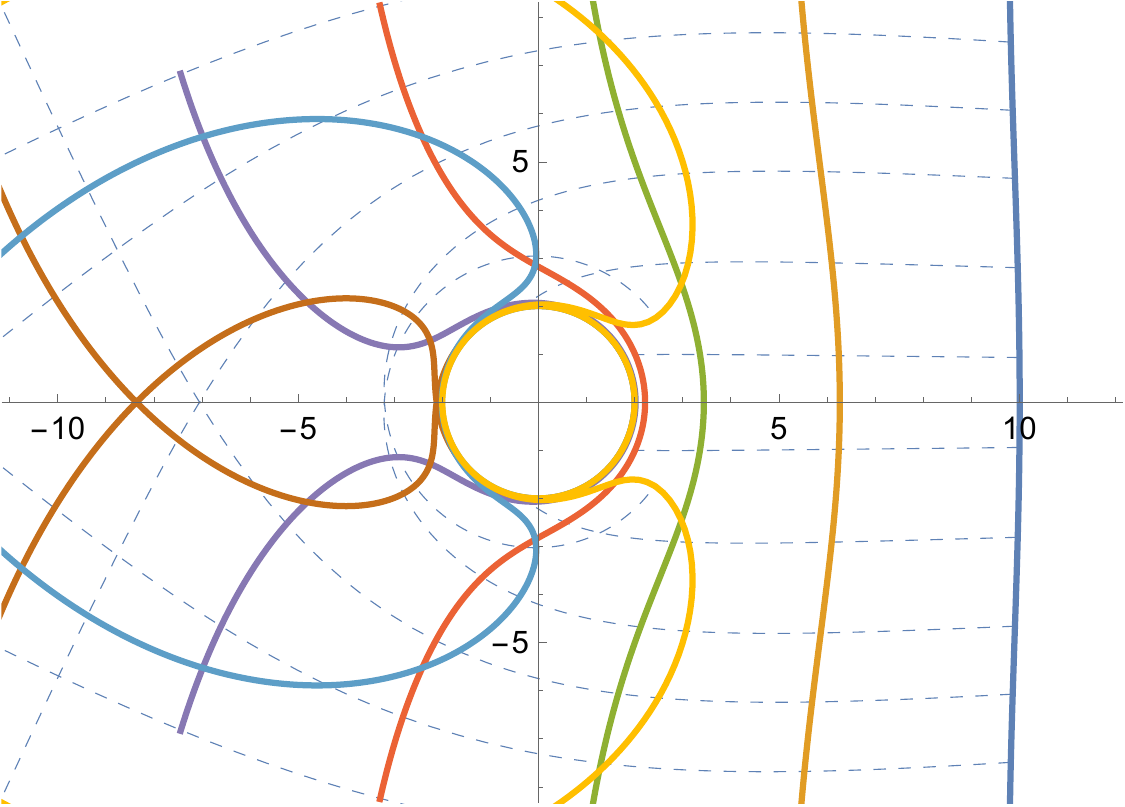} 
        \caption{Infinitely large lightspheres are depicted as colorful curves in four-dimensional Schwarzschild black hole spacetime. Regions on the right side of the lightspheres are generalized Rindler wedges. When lightspheres self-intersect, only the near-horizon regions surrounded by them are Rindler-convex.
        } \label{LS_Sch}
    \end{figure}

After we have obtained the explicit solutions of photon trajectories and infinitely large lightspheres, we have already found a set of Rindler-convex surfaces, which are these infinitely large lightspheres, as shown in Figure \ref{LS_Sch}. The GRWs are the subregions on the side of these lightspheres without the black hole horizon. Besides these Rindler convex surfaces, using theorem \ref{CAPtheorem}, we can choose the intersection of the infinite lightspheres to generate more Rindler convex surfaces and GRWs.

\section{Relationship of two versions of multipartite squashed entanglement}
\noindent
This appendix establishes the equivalence of the geometric conditions for both forms of multipartite squashed entanglement to vanish, ensuring consistency in our analysis of entanglement structures. {Given the fact that the equivalence of these measures in quantum information theory is proven \cite{Wilde_2015}, this coincidence strongly suggests a bidirectional correspondence between geometric structure and squashed entanglement, rather than a one-way derivation.}


    To demonstrate the equivalence between the geometric conditions corresponding to the two versions of multipartite squashed entanglement in gravitational systems, we provide a purely geometric proof. The first step is to prove that the intersection ($E$) of two Rindler-convex regions $AE$ and $BE$ must also be Rindler-convex.

    Due to the tangential condition of Rindler-convexity, any lightsphere that is externally tangential to a Rindler-convex region cannot penetrate its interior. In the case of the intersection ($E$) of two Rindler-convex regions $AE$ and $BE$, the surface of $E$ can be divided into two parts: one part coincides with the surface of $AE$, and the other part coincides with the surface of $BE$. Since the lightspheres that are externally tangential to these two parts cannot reach the interior of $AE$ and $BE$ respectively, they also cannot reach the interior of $E$. This implies that $E$ is Rindler-convex.

    As depicted in figure \ref{trisep}, the geometric conditions that correspond to $\widetilde{E}_{\mathrm{sq}}=0$ are that regions $ABE$, $BCE$, $ACE$, and $ABCE$ are all Rindler-convex. Therefore, the intersections of these regions, namely regions $AE$, $BE$, and $CE$, are also Rindler-convex. This satisfies the geometric condition that corresponds to $E_{\mathrm{sq}}=0$.
    
    Indeed, the equivalence between the geometric conditions corresponding to $E_{\mathrm{sq}}=0$ and $\widetilde{E}_{\mathrm{sq}}=0$ can be demonstrated. If regions $AE$, $BE$, $CE$, and $ABCE$ are all Rindler-convex, we can show that regions $ABE$, $BCE$, and $ACE$ are also Rindler-convex.

Considering region $ABE$, its surface can be divided into two parts: one part coincides with the surface of $ABCE$, and the other part coincides with the surfaces of $AE$ and $BE$. By the Rindler-convexity of these regions, any lightsphere externally tangential to these two parts cannot reach the inside of region $ABE$. Thus, region $ABE$ satisfies the definition of Rindler-convexity.
By applying the same reasoning, we can conclude that region $BCE$ and region $ACE$ are also Rindler-convex. Therefore, the geometric conditions corresponding to $E_{\mathrm{sq}}=0$ and $\widetilde{E}_{\mathrm{sq}}=0$ are equivalent.
This proof can be generalized to the n-partite case in a similar manner.

    
\bibliography{reference}
\bibliographystyle{JHEP}

\end{document}